\documentclass{article}

\title{Discrepancy-based Inference for Intractable Generative Models using Quasi-Monte Carlo}

\author{Ziang Niu$^{1,*}$, Johanna Meier$^{2,*}$, Fran\c{c}ois-Xavier Briol$^{3,\dagger}$\vspace{3mm} \\ 
{\small
$^1$University of Pennsylvania,
$^2$Leibniz Universit\"{a}t Hannover, 
$^3$University College London,} \\ 
{\small
$^*$contributed equally,  
$^\dagger$corresponding author.}}

\usepackage[style=numeric,giveninits=true,maxbibnames=99,url=false,doi=false,isbn=false]{biblatex}
\usepackage{mathrsfs}
\usepackage{algorithm}
\usepackage{algorithmic}
\usepackage{amsthm}
\usepackage{amsmath}
\usepackage{bm}
\usepackage{comment}
\usepackage{xcolor}
\usepackage{amssymb}
\usepackage{geometry}
\usepackage{appendix}
\usepackage{graphicx}
\usepackage{multicol}
\usepackage[colorlinks,linkcolor=blue,anchorcolor=blue,citecolor=blue]{hyperref}

\def\mMMD{\mathrm{MMD}}

\def\VHK{V_{\text{HK}}}

\def\X{\mathcal{X}}
\def\PX{\mathcal{P}(\mathcal{X})}
\def\PY{\mathcal{P}(\mathcal{Y})}

\def\Pk{\mathcal{P}_k(\mathcal{X})}
\def\Hk{\mathcal{H}_k}

\def\F{\mathcal{F}}
\def\E{\mathbb{E}}
\def\P{\mathbb{P}}
\def\Q{\mathbb{Q}}
\def\R{\mathbb{R}}
\def\N{\mathbb{N}}
\def\U{\mathbb{U}}

\def\lf{\left\lfloor}   
\def\rf{\right\rfloor}

\newenvironment{talign*}
 {\csname align*\endcsname}
 {\endalign}
\newenvironment{talign}
 {\csname align\endcsname}
 {\endalign}

\def\d{d}

\newtheorem{assumption}{Assumption}

\newtheorem{proposition}{Proposition}
\newtheorem{theorem}{Theorem}
\newtheorem{lemma}{Lemma}
\newtheorem{corollary}{Corollary}

\begin{document}

\maketitle


\begin{abstract}
    Intractable generative models are models for which the likelihood is unavailable but sampling is possible. Most approaches to parameter inference in this setting require the computation of some discrepancy between the data and the generative model. This is for example the case for minimum distance estimation and approximate Bayesian computation. These approaches require sampling a high number of realisations from the model for different parameter values, which can be a significant challenge when simulating is an expensive operation. In this paper, we propose to enhance this approach by enforcing ``sample diversity" in simulations of our models. This will be implemented through the use of quasi-Monte Carlo (QMC) point sets. Our key results are sample complexity bounds which demonstrate that, under smoothness conditions on the generator, QMC can significantly reduce the number of samples required to obtain a given level of accuracy when using three of the most common discrepancies: the maximum mean discrepancy, the Wasserstein distance, and the Sinkhorn divergence. This is complemented by a simulation study which highlights that an improved accuracy is sometimes also possible in some settings which are not covered by the theory.
\end{abstract}


\section{Introduction}

A particular challenge for statistics is the growing complexity of phenomena modelled by scientists, and as a result the growing complexity of the models themselves. This can often lead to cases where a closed form of the likelihood is not available anymore. As a result, classical parameter estimation tools such as maximum likelihood estimation or Bayesian inference cannot be used. Within these so-called intractable likelihood models, intractable generative models are parametric families of probability distributions which are specified through a generative process, so that it is possible to obtain realisations for any value of the parameter \cite{Cranmer2020}. These models are widely used throughout the sciences including genetics \cite{Beaumont2002}, astronomy \cite{Cameron2012} and ecology \cite{Beaumont2010}. In machine learning, one of the main applications is for the simulation of realistic looking images \cite{Mohamed2016}; see the recent line of work on generative adversarial networks \cite{Goodfellow2014}. 

Denote by $\P_\theta$ any element of a parametric family of interest with parameter $\theta$, and let $\X$ be the space of realisations from this model. The generative process of $\P_\theta$ is usually summarised through a pair $(\mathbb{U},G_\theta)$ which includes a relatively simple probability distribution $\mathbb{U}$ (such as a Gaussian or uniform) on some space $\mathcal{U}$ and a parametric map $G_\theta:\mathcal{U} \rightarrow \X$ called a generator or simulator. To obtain $n$ independent and identically distributed (IID) realisations $\{x_i\}_{i=1}^n$ from the model for some fixed parameter $\theta$, one can simply sample IID realisations $u_i \sim \mathbb{U}$, then map these samples through the generator $x_i =G_\theta(u_i)$. The main advantage of generative models is that one can model ever more complex phenomena by increasing the flexibility of the generator, as long as the map $G_\theta$ can be evaluated pointwise. 

Since simulating data is the only option available in the case of generative models, many inference methods for this class are based on simulating synthetic data for various parameter values, then comparing the simulated data to the observations to select a ``good'' parameter value. The latter usually requires defining some notion of distance, or discrepancy, between the two datasets. Once a discrepancy is defined, one possible approach is the framework of \emph{minimum distance estimation (MDE)} \cite{Parr1980}, where an estimator is constructed as the minimiser (over the set of model parameters) of the discrepancy between datasets. In the Bayesian literature, an alternative approach called \emph{approximate Bayesian computation (ABC)} \cite{Beaumont2002} consists of constructing a pseudo-posterior distribution over parameters by selecting parameter values simulated from a prior distribution for which the discrepancy between simulated and actual data is small. In all of the cases above, thinking of the actual data as an approximation to the data-generating process of interest, the main computational challenge can be summarised as having to efficiently estimate some discrepancy given access to realisations of two distributions.

{\color{black}There is a vast literature on possible discrepancies, each with competing advantages for parameter estimation including efficiency, robustness to model misspecification, computational cost and sample complexity.}
In this paper, we will not aim to be exhaustive, but will focus on a small subset of discrepancies which are popular in the literature because
they lend themselves to efficient implementations. The first discrepancy is the maximum mean discrepancy ($\mMMD$) \cite{Gretton2006}, which compares embeddings of probability distributions into reproducing kernel Hilbert spaces, and can be straightforwardly computed through evaluations of a kernel. This was studied by \cite{Briol2019,CheriefAbdellatif2020,CheriefAbdellatif2019,Alquier2020,Dellaporta2022} in the context of MDE, and by \cite{Li2015,Dziugaite2015,Li2017,Sutherland2017,Binkowski2018} for the case where $G_\theta$ is a neural network in particular. It was also used by \cite{Nakagome2013,Park2016,Mitrovic2016,Kajihara2018,Bharti2020} in the context of ABC. The two other discrepancies we will consider are the Wasserstein distance, as well as its relaxation called the Sinkhorn divergence. These can be efficiently implemented thanks to algorithmic advances in computational optimal transport \cite{Peyre2019}. They were considered for MDE by \cite{Bassetti2006,Bernton2017,Genevay2018,Deshpande2018,Wu2019,Nadjahi2019,Nguyen2020,Shen2020} and for ABC by \cite{Bernton2019,Goffard2020,Nadjahi2020}.

Clearly, any algorithmic development improving our ability to estimate these discrepancies will significantly reduce the overall computational cost of implementing all of the algorithms described above. We propose to tackle this problem through the use of quasi-Monte Carlo (QMC) point sets \cite{Dick2010}. In particular, we focus on the case where $\mathbb{U}$ is a uniform distribution\footnote{The assumption that $\U$ is uniform is relatively minor due to Sklar's theorem, which states that any multivariate distribution can be obtained through a transformation of a uniform distribution.} and replace independent and identical distributed (IID) realisations by some QMC point set. This is a rather simple algorithmic trick, which we will call QMC sampling and which has been explored for a wide range of models; see for example \cite{Cambou2017} for copula models,  or \cite{Hofert2018,Hofert2020} for neural networks. Once again, a full review of QMC sampling is out of scope for this paper. Intuitively, this approach consists of generating a more ``diverse'' set of samples from the model. This can be observed visually through the example in Figure \ref{fig:QMCpointset2D} which compares realisations from a Gaussian distribution obtained through Monte Carlo (MC) and QMC. Clearly, the realisations obtained through QMC provide an improved approximation of $\P_\theta$ in the intuitive sense that they provide a more uniform coverage of areas of high-probability under $\P_\theta$.

The main contribution of this paper is a set of theoretical results demonstrating the advantages of QMC sampling for performing inference with discrepancies. In particular, Theorem \ref{thm:QMC_concentration_ineq}, Theorem \ref{thm:Wasserstein_sample_complexity_QMC} and Theorem \ref{thm:Sinkhorn_QMC} provide sample complexity results with respect to the MMD, Wasserstein and Sinkhorn divergence respectively. In each case, the theorem provides sufficient conditions for estimating the discrepancy at a rate which is linear (up to log factors) in the number of realisations $n$. This is a significant improvement upon the usual MC rate which decreases at a root-$n$ speed. Of course, such speed-ups do come at the cost of the generality of the method as they require certain regularity conditions on $G_\theta$ and $\X$. Despite this drawback, we show through an extensive simulation study that faster rates than MC (although not necessarily linear) can still be obtained for QMC in some settings not covered by our theory. We therefore see this paper as an initial step in the study of the use of QMC sampling for discrepancy estimation.

\begin{figure}[t!]
    \centering
  \includegraphics[width=0.9\textwidth]{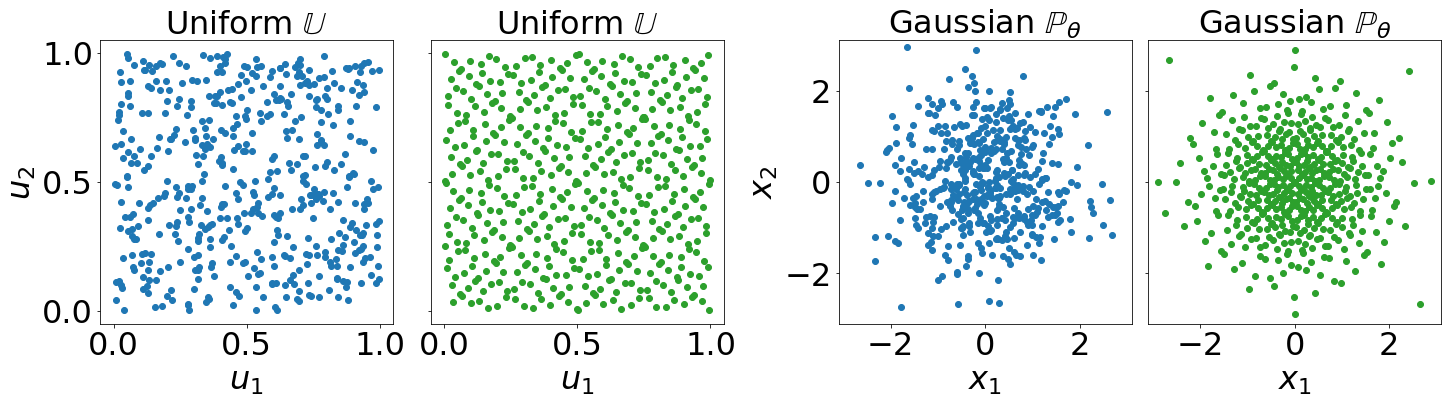}
  \caption{Realisations from $\P_\theta =\mathcal{N}(0, I)$, a zero mean Gaussian with identity covariance matrix. We compare realisations from a $\text{Unif}([0,1]^2)$ (left plot) against a QMC point set (center left plot), together with their projections through the generator $x = G_\theta(u) = (\Phi^{-1}_{\theta_1}(u_1),\Phi^{-1}_{\theta_2}(u_2))$ (center right and right plot), where $\Phi_\theta$ denotes the cumulative distribution function of $\P_\theta$.
  }
  \label{fig:QMCpointset2D}
\end{figure}

The remainder of this paper is structured as follows. In Section \ref{sec:background}, we introduce intractable generative models and the most common distances used for inference, including the MMD, Wasserstein distance and Sinkhorn divergence. In Section \ref{sec:sample_complexity}, we derive our novel sample complexity results. Finally, the performance of the performance of these novel estimators is studied numerically in Section \ref{sec:experiments}. We conclude by discussing potential future directions in Section \ref{sec:conclusions}.


\section{Background}\label{sec:background}

This section will recall background material on inference for intractable generative models (in Section \ref{sec:inference_problem}), then introduce the main discrepancies considered in the literature (in Section
\ref{sec:distances}).

\subsection{Inference for Intractable Generative Models through Discrepancies}\label{sec:inference_problem}

Throughout this paper, we will consider settings where the base space is $\mathcal{U} = [0,1]^s$, the data space satisfies $\mathcal{X} \subseteq \R^d$ and the parameter space satisfies $\Theta \subseteq \R^p$ for $s,p,d \in \mathbb{N}_+ = \{1,2,3,\ldots\}$. We will denote by $\mathcal{P}(\mathcal{X})$ the set of all Borel probability distributions on $\mathcal{X}$. 

The inference task of interest can be summarised as follows. Given IID realisations $\{y_{j}\}_{j=1}^{m}$ from some unknown $\Q \in \mathcal{P}(\X)$, we would like to find the parameter $\theta^{*}\in\Theta$ such that $\P_{\theta^{*}}$ is closest to $\Q$ in some sense. In particular, if $\Q \in \{\P_{\theta}\in\mathcal{P}(\mathcal{X}):\theta\in\Theta\}$ (i.e. the model is well-specified), our task is to recover the parameter value $\theta^*$ which was used to simulate the observations $\{y_{j}\}_{j=1}^{m}$.  One approach is to use a discrepancy, which we will define to be any function $D:\PX \times \PX \rightarrow [0,\infty)$. Specific examples will be provided in Section \ref{sec:distances}, but for now we will only assume such a discrepancy has been selected, and describe how it can be used for inference. Firstly, we may construct an estimator through the framework of MDE \cite{Parr1980}:
\begin{talign*}
\hat{\theta}^D_m \in \arg\min_{\theta\in\Theta}D\left(\P_{\theta},\Q^m\right),
\end{talign*}
 where $\Q^{m}(\d x)=\frac{1}{m}\sum_{j=1}^{m}\delta_{y_{j}}(\d x)$ is an empirical measure, and $\delta_{y_{j}}$ a Dirac measure at $y_{j}$. Of course, this is usually an intractable optimisation problem since it requires evaluating $D$ pointwise at $\P_\theta$, which is itself unknown.  {\color{black}As a result, a common approach is to solve the optimisation problem through evaluations of $D\left(\P^n_{\theta},\Q^{m}\right)$, or of its gradient, where $\P^{n}_\theta(\d x)=\frac{1}{n}\sum_{i=1}^{n}\delta_{x_{i}}(\d x)$  and is obtained from realisations $\{x_i\}_{i=1}^n$ from $\P_\theta$.}
{\color{black} For all discrepancies considered in this paper, $D\left(\P^n_{\theta},\Q^{m}\right)$ is a biased estimate of $D\left(\P_{\theta},\Q^{m}\right)$. This leads to the use of stochastic optimisation methods with biased gradient estimates, which leads to a bias in the estimated parameter \cite{Tadic2017,Karimi2019}. However, any approach leading to more efficient estimation of $D\left(\P_{\theta},\Q^{m}\right)$ may be able to significantly reduce this bias.
}

Secondly, we may use ABC, which aims to construct a pseudo-posterior which closely approximates the exact Bayesian posterior \cite{Beaumont2002}. This can be achieved by sampling parameter values $\{\theta_k\}_{k=1}^K$ (for some $K\in \mathbb{N}$) from a prior distribution $\pi$, then for each of these values simulating a dataset $\{x^k_i\}_{i=1}^n$. Each of these parameter values is then accepted as a realisation from the  pseudo-posterior if $D(\P_{\theta_k}^n,\Q^m) \leq \varepsilon$ holds for some threshold parameter $\varepsilon >0$. This straightforward procedure allows us to sample from the following pseudo-posterior:
\begin{talign*}
    \Pi^{D}_{\varepsilon}(d\theta|y_1,\ldots,y_m) \propto \Pi(d\theta) \mathbb{E}\left[\bm{1}_{\left\{D\left(\P^n_{\theta},\Q^{m}\right) \leq \varepsilon \right\}}(d\theta)\right],
\end{talign*}
where $\bm{1}_A$ is an indicator function for the event $A$, and the expectation is with respect to the randomness in the simulated data. {\color{black}Note that this sampling procedure is only necessary due to the intractability of $D(\P_{\theta_k},\Q^m)$ for intractable generative models; if this quantity was tractable, we would instead want to verify whether $D(\P_{\theta_k},\Q^m) \leq \varepsilon$ instead of $D(\P_{\theta_k}^n,\Q^m) \leq \varepsilon$.}

QMC has previously been used for ABC \cite{Buchholz2019}, but this was used to improve sampling of parameters instead of simulating the data. Finally, we also point out that recent generalised Bayesian procedures for generative models are also discrepancy-based; see for example \cite{Schmon2020,Pacchiardi2021}.

Clearly MDE and ABC critically rely on $D\left(\P^n_{\theta},\Q^{m}\right)$ approaching $D\left(\P_{\theta},\Q^{m}\right)$ at a fast rate in $n$. Whether this is possible will depend on the discrepancy $D$.


\subsection{Examples of Discrepancies for Inference}\label{sec:distances}

Recall that any discrepancy $D$ such that $\forall \P_1,\P_2,\P_3 \in \PX$: (i) $D(\P_1,\P_2)=0$ if and only if $\P_1=\P_2$, (ii) $D(\P_1,\P_2)=D(\P_2,\P_1)$, and (iii) $D(\P_1,\P_2)\leq D(\P_1,\P_3)+D(\P_3,\P_2)$, is called a probability metric on $\PX$. If only (i) holds, $D$ is called a (statistical) divergence. The discrepancies in this paper closely relate to integral probability metrics (IPMs) \cite{Muller1997}. Given a set of functions $\mathcal{F}$, an IPM is a probability metric which takes the form:
\begin{talign*}
    D_{\mathcal{F}}(\P,\Q):=\sup_{f\in\mathcal{F}}\left|\int_{\mathcal{X}}f(x)\P(\d x)-\int_{\mathcal{X}}f(x)\Q(\d x)\right|.
\end{talign*}
In practice, $\mathcal{F}$ needs to be large enough to be able to differentiate $\P$ from $\Q$, but also small enough so that $D_{\mathcal{F}}(\P,\Q)$ can be computed, or at least approximated up to high accuracy. It should also not be too large since we might otherwise have $D_{\mathcal{F}}(\P,\Q) = \infty$ for all $\P \neq \Q$. This can significantly restrict the choices of $\F$ available for inference. In the case where $\mathcal{F}$ is finite-dimensional, the discrepancy above can be thought of as comparing a finite number of summary statistics of $\P$ and $\Q$, as commonly done for the method of simulated moments or in ABC. For this case, the use of QMC was previously studied in \cite{Forneron2019}. In contrast, our work will focus on the most common discrepancies based on infinite-dimensional $\mathcal{F}$, which we introduce below. 

\paragraph{Maximum Mean Discrepancy} Let $\mathcal{F} = \{f:\mathcal{X}\rightarrow \R : \|f\|_{\Hk} \leq 1\}$, the unit-ball of a reproducing kernel Hilbert space (RKHS) $\mathcal{H}_{k}$ with kernel $k:\mathcal{X}\times\mathcal{X}\rightarrow\R$. In this case, the IPM is called the (kernel) \emph{maximum mean discrepancy} \cite{Gretton2006}.  We will assume that the kernel is characteristic, which guarantees that the discrepancy is a metric on the set
\begin{talign*}
\mathcal{P}_{k}(\mathcal{X}) := \{ \P \in \PX : \int_{\mathcal{X}}\sqrt{k(x,x)}\P(\d x)<\infty\} \subseteq \PX,
\end{talign*}
see \cite{Sriperumbudur2009} for more details. The name MMD originates from the fact that the IPM can be expressed as $\mMMD(\P,\Q)=\|\int_\X k(\cdot,y) \P(\d y) - \int_\X k(\cdot,y) \Q(\d y) \|_{\mathcal{H}_{k}}$, which is the size of the difference between $\P$ and $\Q$ when embedded in $\mathcal{H}_k$. The squared-$\mMMD$ can alternatively be expressed as
\begin{talign}
    \mMMD^{2}(\P,\Q)
    & :=\int_{\mathcal{X}}\int_{\mathcal{X}}k(x,y)\P(\d x)\P(\d y)-2\int_{\mathcal{X}}\int_{\mathcal{X}}k(x,y)\P(\d x)\Q(\d y) \nonumber\\
    &\qquad +\int_{\mathcal{X}}\int_{\mathcal{X}}k(x,y)\Q(\d x)\Q(\d y). \label{eq:defn_MMD}
\end{talign}
Note that this expression does not require the computation of a supremum anymore.
Given two empirical measures $\P^n = \frac{1}{n}\sum_{i=1}^n \delta_{x_i}$ and $\Q^m = \frac{1}{m} \sum_{j=1}^m \delta_{y_j}$ approximating $\P$ and $\Q$ respectively, this expression lends itself naturally to the following approximation:
\begin{talign}
    \mMMD^{2}(\P^{n},\Q^{m})
    & =\frac{\sum_{i\neq j}^{n}k(x_i,x_j)}{n^2}-\frac{2\sum_{i=1}^{n}\sum_{j=1}^{m}k(x_i,y_{j})}{nm}+\frac{\sum_{i\neq j}^{m}k(y_{i},y_{j})}{m^2}. \label{eq:MMD_Vstat_defn}
\end{talign}
The use of a U-statistic may also be preferred in some case; see for example \cite{Briol2019}.
One of the main advantages of the MMD is the fact that it can be easily approximated, but also that the kernel choice allows for significant flexibility. The most common example is the Gaussian (or squared-exponential) kernel $k(x,x')= \lambda^2 \exp\left(-\|x-x'\|_{2}^{2}/\sigma^{2}\right)$ where $\lambda,\sigma>0$. QMC point sets were already used with the MMD in \cite{Hofert2018,Hofert2020} in the context of neural network generators, but those papers do not study the sample complexity of the approach from a theoretical viewpoint.

\paragraph{Wasserstein Distance} Let $c:\X \times \X \rightarrow [0,\infty)$ be a metric (called cost function), $p \geq 1$ and $\Gamma(\P,\Q) \subset \mathcal{P}(\mathcal{X} \times \mathcal{X})$ be the set of distributions  with marginals $\P \in \PX$ and $\Q \in \PX$ in the first and second coordinate respectively. The Wasserstein distance can be expressed as:
\begin{talign*}
    W_{c,p}(\P,\Q) & :=  \left(\min\limits_{\gamma\in \Gamma(\P,\Q)}\int_{\mathcal{X}\times\mathcal{X}}c^{p}(x,y)\gamma(dx,dy) \right)^{\frac{1}{p}}, 
\end{talign*}
A common choice for $c$ is the Euclidean distance, but other metrics can be used. The Wasserstein distance is a probability metric on the set 
\begin{talign*}
\mathcal{P}_{c,p}(\mathcal{X}) = \{\P \in \PX : \int_\X c^p(x,y) \P(\d x) < \infty \; \forall y \in \X \} \subseteq \PX.
\end{talign*}
 Although computing the Wasserstein distance for general $\P$ and $\Q$ is usually not possible, it is straightforward to do so for empirical measures $\P^n$ and $\Q^m$ (see for example Chapter 3 in \cite{Peyre2019}):
\begin{talign*}
   W_{c,p}&(\P^n,\Q^m) = \left(\min\limits_{P} \sum_{i=1}^n \sum_{j=1}^m c^p(x_i,y_j) P_{ij}\right)^{\frac{1}{p}},
\end{talign*}
where the minimisation is performed over all $n \times m$ matrices such that $P_{ij} \neq 0$ $\forall i,j$, $\sum_{i=1}^n P_{ij} = \frac{1}{m}$ and $\sum_{j=1}^m P_{ij} = \frac{1}{n}$.
To approximate $W_{c,p}(\P,\Q^m)$, a natural approach is to use $W_{c,p}(\P^n,\Q^m)$, but this is known to have a slow convergence rate as $n$ increases whenever $d>1$ \cite{Fournier2015}. In the special case where $p=1$, the Wasserstein distance is an IPM which corresponds to taking $\mathcal{F}$ to be the set of functions with Lipschitz constant $1$: $\{f:\X \rightarrow \R \text{ s.t. }\forall x,y \in \X,  |f(x)-f(y)|\leq c(x,y)  \}$. This is therefore another setting of infinite-dimensional $\F$ where the supremum does not need to be computed numerically.

\paragraph{Sinkhorn Divergence} A common relaxation of the Wasserstein distance is the following: 
\begin{talign*}
    \bar{W}_{c,p,\lambda}(\P,\Q) & := \min\limits_{\gamma\in \Gamma(\P,\Q)}\int_{\mathcal{X}\times\mathcal{X}}c^{p}(x,y)\gamma(dx,dy) +\lambda H(\gamma\|\P\otimes\Q)\nonumber, \\
    H(\gamma\|\P\otimes\Q) & :=\int_{\mathcal{X}\times\mathcal{X}}\log\left(\frac{\gamma(dx,dy)}{\P(dx)\Q(dy)}\right)\gamma(dx,dy)
\end{talign*}
where $H(\pi \|\P\otimes\Q)$ is called the relative entropy, and $\P\otimes\Q$ is the product measure. Since this discrepancy is not normalised, it is common to work instead with the \emph{Sinkhorn divergence} \cite{Genevay2019}:
\begin{talign*}
    S_{c,p,\lambda}(\P,\Q)=\bar{W}_{c,p,\lambda}(\P,\Q)-\frac{1}{2}\left(\bar{W}_{c,p,\lambda}(\P,\P)+\bar{W}_{c,p,\lambda}(\Q,\Q)\right),
\end{talign*}
which guarantees the resulting value is greater or equal to zero. The Sinkhorn divergence is also symmetric, but does not satisfy the triangle inequality and so is not a metric. However, it does interpolate between the two IPMs we have seen so far: as $\lambda \rightarrow 0$, $S_{c,p,\lambda}(\P,\Q) \rightarrow W_{c,p}(\P,\Q)$, whereas when $\lambda \rightarrow \infty$, $S_{c,p,\lambda}(\P,\Q) \rightarrow \text{MMD}(\P,\Q)$ with kernel $k=-c$ \cite{Feydy2019}. Once again, it is straightforward to compute $S_{c,p,\lambda}(\P^n,\Q^m)$ in the case of empirical measures, and this can be used to estimate the exact Sinkhorn divergence: $S_{c,p,\lambda}(\P,\Q^m)$.
{\color{black}
From a computational viewpoint, one particular advantage of the Sinkhorn divergence over the Wasserstein distance is that it has better sample complexity when using Monte Carlo points in multiple dimensions \cite{Genevay2019}. We will return to this point in the next section on QMC sample complexity.}

\paragraph{Sliced Discrepancies} A final example of discrepancies commonly used for inference are the so-called \emph{sliced discrepancies} \cite{Kolouri2020}. The main motivation for these is to construct discrepancies which will be useful for high-dimensional problems. This is done by projecting probability distributions on $\X$ to probability distributions on some lower dimensional space $\mathcal{Y}$ (usually one dimension) using a map $\mathcal{S}_{\xi}: \PX \rightarrow \PY$, then comparing these projections using any discrepancy $D: \PY \times \PY \rightarrow [0,\infty)$, such as those discussed above. The corresponding sliced discrepancy consists of an average over possible projections:
\begin{talign*}
SD(\P,\Q) = \int_{\Xi} D(\mathcal{S}_\xi \P, \mathcal{S}_\xi \Q) d\xi,
\end{talign*}
where $\mathcal{S}_\xi \P, \mathcal{S}_\xi \Q$ are the projections of $\P,\Q$ along the direction $\xi \in \Xi$, and $\Xi$ is the space of directions considered. In order to compute the discrepancy, an MC estimator is used: $\widehat{SD}(\P,\Q) = \frac{1}{L} \sum_{l=1}^L D(\mathcal{S}_{\xi_l} \P, \mathcal{S}_{\xi_l} \Q)$ where $\{\xi_l\}_{l=1}^L$ are MC realisations from a uniform distribution over $\Xi$. The most common sliced-discrepancy is the sliced-Wasserstein distance $SW_{c,p}$ \cite{Deshpande2018,Wu2019,Nadjahi2019,Nguyen2020}, in which case $D$ is $W_{c,p}$ and the projections are constructed using the Radon transform.


\section{Sample Complexity with Quasi-Monte Carlo}
\label{sec:sample_complexity}

Now that we have introduced the main discrepancies which will be considered in this paper, we are ready to introduce our novel sample complexity results based on QMC and RQMC. We first introduce the methodology in Section \ref{sec:QMCgenerator}, then provide theoretical results demonstrating improved sample complexity for MMD in Section \ref{sec:sample_complexity_MMD} and for the Wasserstein distance and its Sinkhorn approximation in Section  \ref{sec:sample_complexity_Wasserstein} and \ref{sec:sample_complexity_sinkhorn} respectively. These results all build upon the work of \cite{Basu2016}, which considered the use of QMC for integrating compositions of functions.

\paragraph{Notation} For two sequences $\{f_n\}_{n \in \N}$ and $\{g_n\}_{n \in \N}$,  $f_n = O(g_n) \Leftrightarrow \limsup_{n \rightarrow \infty} |f_n/g_n| < \infty$. For some $f:\mathcal{X} \rightarrow \mathbb{R}$ and multi-index $\alpha =(\alpha_1,\ldots,\alpha_d) \in \mathbb{N}^d$, we will denote by $\partial^\alpha f$ the partial derivative $\partial^{|\alpha|}f/\partial^{\alpha_1} x_1 \ldots \partial^{\alpha_d}x_d$. The space $\mathcal{C}^{m}(\X)$ of $m$-continuously differentiable functions ($m \in \mathbb{N}$) corresponds to functions such that $\partial^\alpha f$ is continuous $\forall \alpha \in \mathbb{N}^d$ such that $|\alpha|=\alpha_{1}+\cdots+\alpha_{d} \leq m$. Similarly, $\mathcal{C}^{m,m}(\X \times \X)$ will denote functions $f:\X \times \X \rightarrow \R$ such that $\partial^{\alpha,\alpha}f$ exists and is continuous $\forall \alpha \in \mathbb{N}^d$ with $|\alpha|\leq m$. {\color{black} Relatedly, if we have a set $\beta \subseteq 1:d$, we write $\partial_{\beta}$ to denote the (first-order) mixed partial derivatives of $f$ with respect to the coordinates in the set $\beta$}. Finally, we will write $L^{p}(\X)$ to denote the $p$-integrable functions; i.e. $f:\X \rightarrow \R$ satisfying $\|f\|_{L^p(\X)} := (\int_{\X} f^p(x) dx)^{\frac{1}{p}} < \infty$ (where we will use the common abuse of terminology to avoid technicalities with equivalence classes).

\subsection{Enhancing Sample Diversity through quasi-Monte Carlo}\label{sec:QMCgenerator}

Recall that to obtain realisations $\{x_i\}_{i=1}^n$ from $\P_\theta$, the generative approach consist of obtaining realisations $\{u_i\}_{i=1}^n \sim \text{Unif}([0,1]^s)$, then mapping these through the generator: $x_i = G_{\theta}(u_i)$. Under sufficient regularity conditions on $G_\theta$, we would expect two realisations $x_1,x_2$ to be far from one another whenever $u_1,u_2$ are also far from one another. The main idea in this paper is that we may improve sample diversity by selecting $\{u_i\}_{i=1}^n$ according to a QMC point set. This notion of diversity is usually measured through the \emph{star-discrepancy} of a point set:
\begin{talign*}
    D^{*}(\{u_i\}_{i=1}^n):=\sup_{v\in[0,1]^{s}}\left|\frac{1}{n} \sum_{i=1}^n \bm{1}_{ [0,v)}(u_i)-\prod_{i=1}^{s}v_{i}\right|.
\end{talign*}
We will call a point set $\{u_i\}_{i=1}^n$ such that $D^{*}(\{u_i\}_{i=1}^n)=O(n^{-1} (\log{n})^{{\alpha}_s}) $  for some ${\alpha}_s>0$ as $ n\rightarrow\infty $  a \emph{QMC point set}, and $\alpha_s$ will usually depend on the dimensionality $s$ of the domain $\mathcal{U}$. This is also sometimes referred to as a low-discrepancy point set, but we will avoid this terminology to avoid any confusion between discrepancies on probability distributions and the star discrepancy. Popular constructions \cite{Dick2010} include Hammersley point sets, which are based on infinite van der Corput sequence, and can achieve ${\alpha}_s=s-1$. Alternatively, lattice point sets achieve ${\alpha}_2=2$ and ${\alpha}_s=s$ for $s\geq 3$, $(t,m,s)-$nets in base b achieve ${\alpha}_s=s-1$, and the Halton sequence achieves $\alpha_s = s$. 

Bounds on $D^{*}(\{u_i\}_{i=1}^n)$ are particularly useful since they provide bounds on the integration error for an estimate $\frac{1}{n} \sum_{i=1}^n f(u_i)$ of some real-valued function $f:[0,1]^s \rightarrow \R$ whenever it has bounded Hardy-Krause variation, which will be denoted by $\VHK(f)$. Since the notation for the Hardy-Krause variation is rather involved, we refer the reader to Appendix \ref{appendix:background} for details.

Related constructions are the \emph{randomized QMC (RQMC) point sets}, which are sets of points $\{u_i\}_{i=1}^n$ with distribution $\text{Unif}([0,1]^s)$ such that $\exists N,B>0$ such that for $\forall n \geq N$, $D^*(\{u_i\}_{i=1}^n) \leq B (\log n)^{\alpha_s}n^{-1}$ with probability $1$ for some $\alpha_s>0$. The most common approach to construct these consists of ``scrambling'' a QMC point set, which consists of applying random transformations which preserve the low discrepancy structure. This allows those point sets to be used to obtain unbiased estimates of integrals of some functions against $[0,1]^s$. Details on the construction of the scrambled points can be found in Chapter 17 in \cite{Owen2013}. 

{\color{black}
In the remainder, we will provide technical conditions on $\X$ and $G_\theta$ so that for any $D$ amongst the discrepancies previously mentioned and assuming we use $n$ QMC points, we have
\begin{align*}
|D(\P_\theta,\Q^m) - D(\P_\theta^n,\Q^m)| = O(n^{-1} (\log n)^{\alpha_s}).
\end{align*}
This is an improvement on the MC rate for which the rate would be $O(n^{-\frac{1}{2}})$.
Since the cost of generating MC or QMC realisations is linear in the number of samples, a natural approach to balance the error in $n$ and $m$ of estimating $D(\P_\theta,\Q)$ is to take $n$ growing with $\sqrt{m}$. Note however that this optimal scaling is asymptotic and relies on a number of unknown constants dependent on the QMC point set used and the cost of evaluating the generator. This scaling will be studied further in the experiments.
}
{
\color{black}
\subsection{Technical Assumptions}\label{sec:technical_conditions}

Before stating our sample complexity results, we introduce and discuss the assumptions that will be required. Our first assumption concerns the domain of the generator and the point sets:
\begin{assumption}\label{assumption:points}
Given a model $\P_\theta$ with generative process $(\text{Unif}([0,1]^s), G_\theta)$, we assume we have access to $x_i = G_\theta(u_i)$ for $i=1,\ldots,n$
where $\{u_i\}_{i=1}^n \subset [0,1]^s$ form a QMC or RQMC point set for some $\alpha_s > 0$. Furthermore, we write $\P^n_\theta =\frac{1}{n}\sum_{i=1}^n \delta_{x_i}$.
\end{assumption}

This assumption is very mild since it only assumes we can write the generative model in terms of a generator mapping from $[0,1]^s$ (which is always possible due to Sklar's theorem) and that we have access to a QMC or RQMC point set such as those mentioned above. Such point sets are widely available, for example in Python through the packages \texttt{SciPy} \cite{SciPy2020} and \texttt{QMCPy} \cite{Choi2020}.

For the MMD and Sinkhorn divergence results, we will also require a second assumption on the generator. For this, we will use the notation $a_{v}:b_{-v}$ to represent a point $u\in[a,b]^{s}$ with $u_{j}=a_{j}$ for $j\in v$, and $u_{j}=b_{j}$ for $j\notin v$; see Appendix \ref{appendix:background} for more details. 
\begin{assumption}\label{assumptions:generator}
The generator is a map $G_\theta:[0,1]^s \rightarrow \X$ where:
\begin{enumerate}
     \item $\partial^{(1,\ldots,1)} (G_\theta)_j \in \mathcal{C}([0,1]^s)$ for all $j = 1, \ldots, d$. 
    \item $\partial^{v} (G_\theta)_j(\cdot:1_{-v}) \in L^{p_j}([0,1]^{|v|})$ for all $j = 1, \ldots, d$ and  $v \in \{0,1\}^s\setminus(0,\ldots,0)$, where $p_j \in [1,\infty]$ and $\sum_{j=1}^d p_j^{-1} \leq 1$.
\end{enumerate}
\end{assumption}
Assumption 2.1 is fairly straightforward and simply requires that the mixed partial derivative of the generator with respect to each coordinate is a continuous function, which is usually a condition which should be easy to verify (this needs to be done on a case-by-case basis). For example, in the case of neural network-based generators, the chain rule guarantees that this assumption will be satisfied whenever the activation functions are smooth enough. This is for example the case for the logistic, hyperbolic tangent, Gaussian, softplus and softmax activation functions which are all infinitely differentiable. However, neural generators with less regular activation functions such as the rectified linear unit will not satisfy the condition.  

Assumption 2.2 requires certain integrability conditions for derivatives of the generator. When $\X$ is compact, it follows directly from the first condition. However, this is not true when $\X$ is not bounded and the requirement that $\sum_{j=1}^d p_j^{-d} \leq 1$ is slightly harder to satisfy in that case, especially for high dimensional problems. One straightforward, but rather restrictive, way of guaranteeing the condition is to enforce that derivatives of the form $\partial^{v} (G_\theta)_j(\cdot:1_{-v})$ are all bounded. Alternatively, we could require that $\partial^{v} (G_\theta)_j \in L^{p_j}([0,1]^{|v|})$ for all $j = 1, \ldots, d$ and $v \in \{0,1\}^s\setminus(0,\ldots,0)$, where $p_j \in [1,\infty]$ and $\sum_{j=1}^d p_j^{-1} \leq 1/2$; see Corollary 7 of \cite{Basu2016} for a more detailed discussion. This holds for example when the generator has bounded derivatives. 
}

\subsection{Sample Complexity for Maximum Mean Discrepancy}\label{sec:sample_complexity_MMD}

We are now ready to present our sample complexity results. Our first set of results will provide sufficient conditions on $k$ and $G_\theta$ to guarantee improved sample complexity by the use of (R)QMC point sets. We say that a kernel $k$ is bounded if  $\exists C>0$ such that $\sup_{x,x' \in \mathcal{X}} |k(x,x')| \leq C$. Before presenting this result, we briefly recall a result using IID samples which will be used as a reference. 
\begin{proposition}[Lemma 1 in \cite{Briol2019}]\label{prop:IID_Concentration_ineq}
Assume that $k$ is bounded and let $\P \in \Pk$. Let $\P^{n} = \frac{1}{n} \sum_{i=1}^n \delta_{x_i}$ where $\{x_i\}_{i=1}^n$ are IID realisations from $\P$. Then, with probability $1-\delta$:
\begin{talign*}
\mMMD(\P,\P^{n}) = O(n^{-\frac{1}{2}}) \sqrt{\log(\delta^{-1})}.
\end{talign*}
\end{proposition}
We also only provide a simplified version of the statement which does not make the constants explicit for simplicity. It is also possible to obtain similar results for convergence of the MMD in the case of dependent realisations; see \cite{CheriefAbdellatif2019}. Although the rate in $n$ is independent of dimensions, we will require a large number of samples in order to converge to zero due to the small exponent. The original statement is valid for finite $n$, but we present it in this asymptotic form for ease of comparison with the QMC/RQMC result below.

We now present a new sample complexity for $\mMMD$ using QMC sequences. To do so, we need to show that the space of functions of the form $f \circ G_\theta$ for $f \in \Hk$ is continuously embedded into a space for which QMC can provide fast convergence rates. This is a challenging task, as was highlighted by \cite{Li2020}, and we provide an auxiliary theorem for this (Theorem \ref{thm:general_Sbolev_composition}) in Appendix \ref{appendix:proof_preliminary}. For this theorem to hold, we show that sufficient conditions can be obtained by ensuring that the generator $G_\theta$ and domain $\X$ are regular enough. 

{\color{black}
\begin{theorem}\label{thm:QMC_concentration_ineq}
Let $\P_\theta \in \Pk$ and suppose Assumption \ref{assumption:points} and \ref{assumptions:generator} hold. 
Further assume that $k \in \mathcal{C}^{s,s}(\X)$ and $\forall t\in\mathbb{N}_0^{d},|t|\leq s,\sup_{x\in\X}\partial^{t,t}k(x,x)<C_k$ where $C_k$ is some universal constant depending only on $k$.
Then, 
    \begin{talign*}
    \mMMD(\P_{\theta},\P_{\theta}^{n})=O(n^{-1}{\color{black}(\log{n})^{{\alpha}_s}}). 
    \end{talign*}
\end{theorem}
A direct implication is the following corollary, which follows from the triangle inequality.
\begin{corollary}\label{cor:MMD}
Suppose the conditions in Theorem \ref{thm:QMC_concentration_ineq} hold. Then,
\begin{align*}
    \left|\mMMD(\P_\theta,\Q^{m})-\mMMD(\P^{n}_\theta,\Q^{m})\right|=O(n^{-1}(\log{n})^{{\alpha}_s}). 
\end{align*}
\end{corollary}
}

The proof is in Appendix \ref{appendix:proof_MMD}. When using QMC, our result is only valid asymptotically in $n$, whereas for MC the result is also valid for finite $n$, although it only holds with probability $1-\delta$. When using a RQMC point set the result above holds with probability $1$ for finite but large enough $n$. As compared to Proposition \ref{prop:IID_Concentration_ineq}, this theorem requires additional regularity from the generator (as per Assumption \ref{assumptions:generator}), but also smoothness for $k$. It does however provide a significantly faster convergence rate. The smoothness condition for the kernel is always satisfied in the case of the Gaussian kernel since it is infinitely differentiable; see Section 4.4. of \cite{Steinwart2008}. Note that Theorem \ref{thm:QMC_concentration_ineq} has direct implications for the work of \cite{Hofert2018,Hofert2020}, which considered the use of QMC sampling in the context of MMD generative adversarial networks.

\subsection{Sample Complexity for the Wasserstein Distance} \label{sec:sample_complexity_Wasserstein} 

The main competitor to $\mMMD$ for inference in generative models is the Wasserstein distance. An interesting question is therefore whether QMC can also lead to improved sample complexity results in this setting. We first recall a result for the case of MC realisations. Extensions of this result to dependent realisations can also be found in \cite{Fournier2015}. 
\begin{proposition}[Theorem 1 in \cite{Fournier2015}, simplified] \label{prop:fournier}
Let $\X = \R^d$, $p> 0 $, $c$ be a metric and $\P \in\mathcal{P}_{c,q}(\mathcal{X})$ for $q > p$. Let $\P^n$ be the empirical measure obtained from $n$ IID realisations of $\P$. Then,
\begin{talign*}
    \E\left[W_{c,p}^p(\P,\P^n)\right]=
\begin{cases}
O \left(n^{-\frac{1}{2}}+n^{-\frac{(q-p)}{q}}\right) & \text{ if } p>\frac{d}{2} \text{ and } q \neq 2p.\\
O \left(n^{-\frac{1}{2}}\log(1+n)+ n^{-\frac{(q-p)}{q}}\right)  & \text{ if } p=\frac{d}{2} \text{ and } q \neq 2p.\\
O\left(n^{-\frac{p}{d}} + n^{-\frac{(q-p)}{q}}\right)  & \text{ if } p \in \left[1,\frac{d}{2}\right) \text{ and } q \neq \frac{d}{(d-p)}.
\end{cases}
\end{talign*}
\end{proposition}
The result above is in expectation, but leads directly to a result in probability using Markov's inequality. 
{\color{black} This result shows a significant disadvantage of using the Wasserstein distance for inference in generative models from a computational viewpoint}: it suffers from the curse of dimensionality when $p$ is small relative to $d$ (the scenario most common in practice).
Indeed, in the third case considered above the $n$ required to estimate the distance accurately increases exponentially quickly with $d$. 

The case most commonly considered in practice for inference in generative models is $p=1$ (see for example \cite{Bernton2017,Bernton2019}), in which case the rate is $O(n^{-1/2})$ if $d=1$, $O(n^{-1/2}\log(1+n))$ if $d=2$, and $O(n^{-1/d})$ for $d \geq 3$.  In the next result, we derive a novel result to show the impact of the use of QMC point sets to estimate the Wasserstein distance when $q=1$, in which case the Wasserstein is an IPM. The proof is in Appendix \ref{appendix:proof_Wasserstein}.
{\color{black}
\begin{theorem}\label{thm:Wasserstein_sample_complexity_QMC}
Let $\P_\theta \in \mathcal{P}_{c,1}(\X)$ where $c(x,y)=\|x-y\|$ for some norm $\|\cdot\|$ on $\X \subseteq \R^d$. Suppose that Assumption \ref{assumption:points} holds with $s=d=1$, and assume that $\VHK(G_\theta) < \infty$. Then,  
\begin{talign*}
W_{c,1}(\P_\theta,\P_\theta^n) = O(n^{-1}{\color{black}(\log{n})^{{\alpha}_s}}).
\end{talign*}
\end{theorem}
}

{\color{black} Since our goal is to approximate $W_{c,1}(\P_\theta,\Q^{m})$ with $W_{c,1}(\P_\theta^{n},\Q^{m})$, we also consider:
\begin{corollary}\label{cor:wasserstein}
Suppose the conditions in Theorem \ref{thm:Wasserstein_sample_complexity_QMC} hold. Then,
\begin{align*}
    \left|W_{c,1}(\P_\theta,\Q^{m})-W_{c,1}(\P_\theta^{n},\Q^{m})\right|=O(n^{-1}(\log{n})^{{\alpha}_s}).
\end{align*}
\end{corollary}
}
\color{black}{We note that the assumption that $\VHK(G_\theta) < \infty$ is weaker than that imposed in Assumption \ref{assumptions:generator}, so that the discussion about sufficient conditions also holds here.}
This result shows that the convergence rate can be improved by a $O(n^{-1/2})$ term (up to logarithms) when using a RQMC/QMC point set instead of MC samples in $d=1$ (once again, QMC results are only valid asymptotically). This is significant since in $d=1$, the computational cost for the Wasserstein distance is $O(n \log n)$, which is significantly faster than the $O(n^2)$ cost for the MMD distance. For $d>1$, the optimal rate for approximating an arbitrary distribution with a deterministic point set is $W_{c,1}(\P,\P^n) = O(n^{-1/d})$; see Theorem 2 in \cite{Novak2016}. We therefore cannot hope to obtain an improved sample complexity result in this case.

Fortunately, this is not the end of the story. First, the $d=1$ rate also transfers to sliced-Wasserstein distances in $d>1$ using Theorem 2 in \cite{Nadjahi2020b}. As we will see in the next section, the use of QMC and RQMC for the sliced-Wasserstein distance leads to very favourable computational costs, and warrants further study. Second, the next section will show that the Sinkhorn divergence can also be approximated at a fast rate even for $d>1$.

\subsection{Sample Complexity for the Sinkhorn Divergence} \label{sec:sample_complexity_sinkhorn} 

As for the other discrepancies, we will first review an existing result about the sample complexity of the Sinkhorn divergence with MC samples. Note that the result, which was proved in \cite{Genevay2019}, is in terms of distance between estimated Sinkhorn divergence and the exact Sinkhorn divergence. Results of this form can be obtained from our theorems for the MMD and Wasserstein distance since they are both metrics and hence satisfy the triangle inequality, but here we are working with a divergence instead of a metric and so directly present the result in this form.
\begin{proposition}[Corollary 1 in \cite{Genevay2019}]\label{prop:genevay2019}
Let $\P,\Q \in \mathcal{P}_{c,p}(\X)$ on some bounded $\X \subset \mathbb{R}^{d}$, and suppose $c \in \mathcal{C}^{\infty,\infty}(\X \times \X)$ is a Lipschitz continuous cost function. Let $\P_n$ and $\Q_n$ consist of $n$ IID realisations from $\P$ and $\Q$ respectively. Then, with probability $1-\delta$:
\begin{talign*}
\left|S_{c,p,\lambda}(\P,\Q)-S_{c,p,\lambda}(\P^{n},\Q^{n})\right| = O(n^{-\frac{1}{2}})\sqrt{\log(\delta^{-1})}.
\end{talign*}
\end{proposition}
The constant in this rate depends on $\lambda$ and $d$, and more detailed can be found in Theorem 3 of \cite{Genevay2019}. Most strikingly, the dependence on $\lambda$ is exponential as $\lambda \rightarrow \infty$. See also \cite{Mena2019} for a more refined result when using the squared Euclidean metric as cost function. Given a fixed value of $\lambda$ and $d$, the rate in $n$ is the MC rate. As we will see in the next results, this can be improved upon using QMC/RQMC point sets. {\color{black}
Note here we need to restrict the domain to be compact, which is more restrictive than for our results for the MMD or Wasserstein distance, but is similar to the requirement in \ref{prop:genevay2019}. 

\begin{assumption}\label{assumptions:domain}
Assume that the domain $\X \subset \mathbb{R}^d$ is a compact space.
\end{assumption}
This restriction for the domain is necessary since our proof builds on \cite{Genevay2019}, which requires this assumption to hold. Although there are some results that allow the support of the distribution to be unbounded, for example in \cite{Mena2019} where the compactness assumption was relaxed to distributions with sub-Gaussian tails on unbounded domains, this proof technique require us to enforce stronger regularity conditions for the generator which would limit the applicability of the proof.
}
\begin{theorem}\label{thm:Sinkhorn_QMC}
Let $c\in\mathcal{C}^{\infty,\infty}(\X \times \X)$ and suppose  $\P_{\theta},\Q \in \mathcal{P}_{c,p}(\X)$. Furthermore, suppose Assumptions \ref{assumption:points},  \ref{assumptions:generator} and \ref{assumptions:domain} hold. Then 
\begin{talign*}
    \left|S_{c,p,\lambda}(\P_{\theta},\Q^m)-S_{c,p,\lambda}(\P_{\theta}^{n},\Q^m)\right|=O(n^{-1}{\color{black}(\log{n})^{{\alpha}_s}}).
\end{talign*}
\end{theorem}
The proof is available in Appendix \ref{appendix:proof_sinkhorn}. Note that the rate is now the same as that possible when using QMC/RQMC for the MMD, and it significantly improves on what is possible when working with the Wasserstein distance.



\section{Numerical Experiments} \label{sec:experiments}

In this section, we will return to the uniform and Gaussian models first studied in Figure \ref{fig:QMCpointset2D}, then consider inference for intractable generative models including the multivariate g-and-k distributions, a flexible class of bivariate Beta distributions, \color{black}{and the deep neural network generator of a variational autoencoder}.  The aims of this section are two-fold. First, we will verify that the theoretical results in the previous section hold in practice. Second, we will look at QMC sampling in settings where Assumption  \ref{assumptions:generator} and \ref{assumptions:domain}  are violated. The requirements on the smoothness of $G_\theta$ and the assumption that $\X$ is compact are rather restrictive but necessary to transfer existing theoretical results from the QMC theory to the setting of generative models. Thankfully, we will see that there are  many settings where these assumptions are not satisfied but the approach nevertheless provides significant speed-ups. As such, our paper provides further evidence complementing the extensive discussion of this issue in Chapters 15, 16 and 17 of Art Owen's book \cite{Owen2013}, and opens the way for further extensions of our theoretical results in Section \ref{sec:sample_complexity}.

Our simulation study uses the \texttt{SciPy} \cite{SciPy2020}, \texttt{JAX} \cite{Jax2018}, \texttt{QMCPy} \cite{Choi2020}, \texttt{POT} \cite{flamary2021pot} \color{black}{\texttt{TensorFlow} \cite{Tensorflow2015}} libraries. The code can be found at 
\begin{center}
   \url{https://github.com/johannnamr/Discrepancy-based-inference-using-QMC}.
\end{center}
Unless stated otherwise, all the RQMC results are based on generalised Halton or Sobol sequences which have been randomised using the scrambling factors of \cite{Faure2009}. The approximation of sliced-distances are based on randomly sampled slices as described in Section \ref{sec:background}. Additional results are provided in Appendix \ref{appendix:numerical_experiments}.

\begin{figure}[t!]
    \centering
  \includegraphics[width=\textwidth]{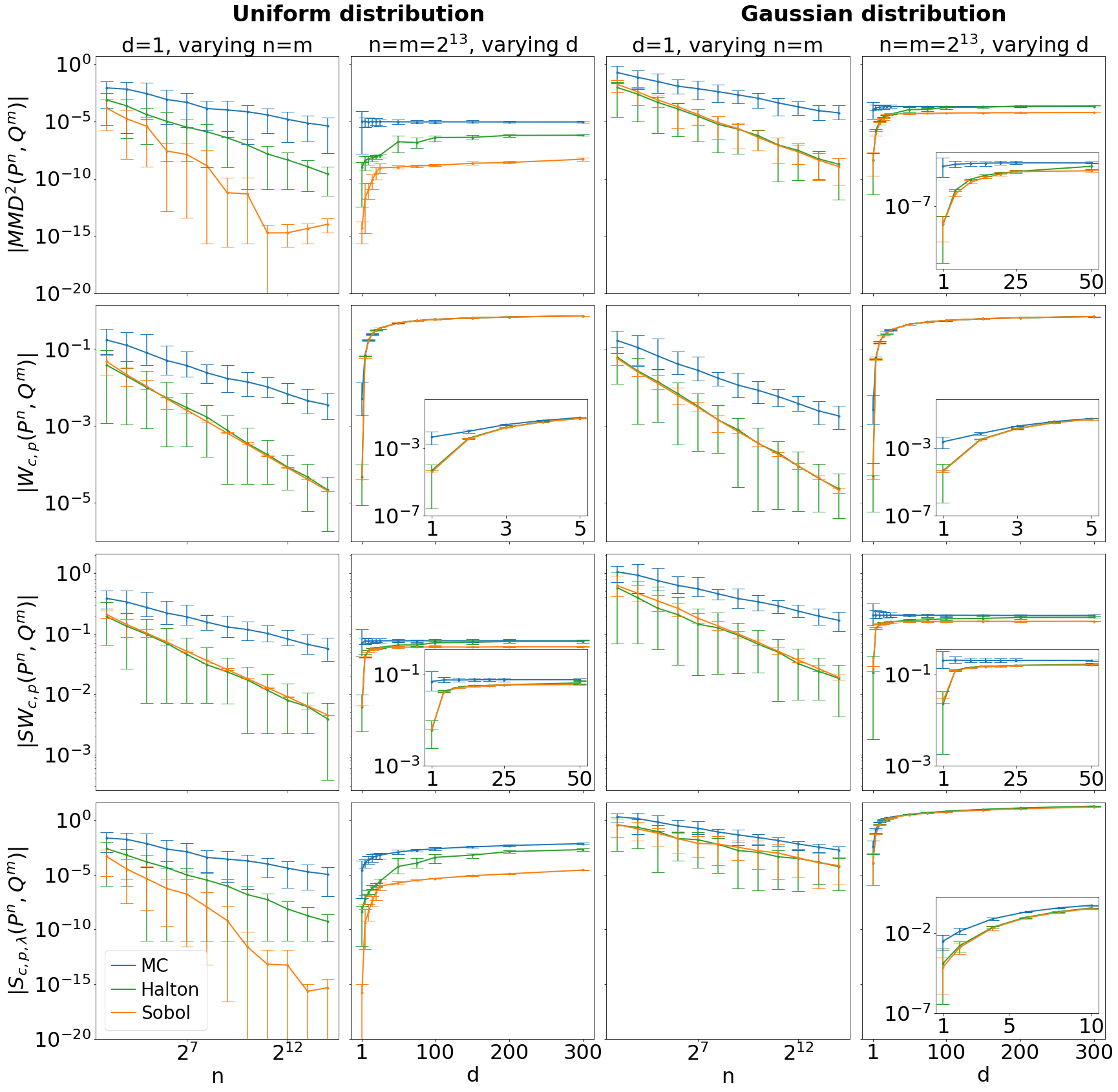}
  \caption{Sample Complexity for a Uniform and Gaussian model in MMD, Wasserstein distance, sliced-Wasserstein distance and Sinkhorn divergence.}
  \label{fig:sample_complexity_unif_gauss}
\end{figure}

\subsection{Sample Complexity for Uniform and Gaussian Models}\label{sec:experiments_unif_gauss}

We first revisit the examples in Figure \ref{fig:QMCpointset2D} which considered uniform and Gaussian distributions. These examples are of course very simple and do not require inference tools for generative models, but their simplicity allows us to study the sample complexity of QMC/RQMC in a wide range of scenarios. For the uniform distribution $\P_\theta = \text{Unif}([0,1]^d)$, we will use $\U = \text{Unif}([0,1]^s)$ with $G_\theta(u)=u$. For the Gaussian distribution $\P_\theta = \mathcal{N}(0,I)$, we use $\U = \text{Unif}([0,1]^s)$ together with the inverse CDF of the univariate standard Gaussian $\Phi$ element-wise: $G_\theta(u) = (\Phi^{-1}(u_1),\ldots,\Phi^{-1}(u_d))$. The simulator $G_\theta$ does not depend on $\theta$ here since we only study the sample complexity results for a fixed distribution. 

For these examples, we have $s=d$, $\P^n = \frac{1}{n}\sum_{i=1}^n \delta_{x_i}$, $\Q^m = \frac{1}{m}\sum_{j=1}^m \delta_{y_j}$ and $n=m$, where $\{x_i\}_{i=1}^n$ and $\{y_j\}_{j=1}^m$ are obtained through the generator $G_\theta$.  Our main results are presented in Figure \ref{fig:sample_complexity_unif_gauss}, and include simulations with MC (in blue), RQMC with Halton sequences (in green) and RQMC with Sobol sequences (in orange). All the experiments have been repeated 25 times. The lines provide the average, and the error bars also represent intervals for the range of values observed. The smaller windows provide a zoomed-in plot for the cases where the gains in performance quickly reduce with $d$.

The first row computes $|\mMMD^2(\P^n,\Q^m)|$ when using a squared-exponential kernel with lengthscale $l=1.5d^{1/2}$. This quantity should decrease as $O(n^{-1/2})$ (see Proposition \ref{thm:QMC_concentration_ineq}) when using MC, and as $O(n^{-1}{\color{black}(\log{n})^{{\alpha}_s}})$
(see Theorem \ref{thm:QMC_concentration_ineq}) when using RQMC. These rates clearly hold for both models when $d=1$, and we see that RQMC quickly provides orders of magnitude improvements as $n$ grows. For the uniform example, the Sobol sequence significantly outperforms the Halton sequence, but this is not the case for the Gaussian example. This is in line with theoretical results showing that the root-mean squared error for Sobol sequences can decrease as
$O(n^{-3/2}{\color{black}(\log{n})^{{\alpha}_s}})$ \cite{Owen2006}, and could motivate further theoretical work extending the results in this paper. Significant improvements are also observed for larger values of $d$, although the gains (if any) are limited for $d>100$ in the Gaussian case. This is not surprising since the Gaussian model does not satisfy the necessary conditions of Theorem \ref{thm:QMC_concentration_ineq} since $\X$ is unbounded and the generator is not sufficiently regular \color{black}{($G_\theta$ is unbounded, and as a result has infinite Hardy-Krause variation; see \cite{Owen2013} Section 15.11)}. The lengthscale is adapted so as to increase with dimension; this is necessary as the distance between points grows exponentially with $d$ due to the curse of dimensionality.

Additional experiments with the Mat\'ern kernel with smoothness parameter $\nu=3/2, 5/2$ and $7/2$ are also provided in Figure \ref{fig:MMD_matern_kernel}. We observe that the performance is significantly improved when using QMC points sets regardless of the choice of kernel, although this advantage decreases when $d$ increases, and is larger for smoother kernels. This is interesting to see since the Mat\'ern kernel does not satisfy the conditions of Theorem \ref{thm:QMC_concentration_ineq} when $d$ is large. Finally, we notice from Figure \ref{fig:MMD_different_pointset} that the results are not very sensitive to the choice of QMC point set.

The second row of Figure \ref{fig:sample_complexity_unif_gauss} illustrates $|W_{c,p}(\P^n,\Q^m)|$ with $c(x,y)=\|x-y\|_2$ and $p=1$. The QMC point sets lead to significant gains when $d=1$, but not for larger $d$ (a small advantage is seen until $d=5$, but this is very limited). Further experiments for alternative choices of $c$ and $p$ can also be found in Figure \ref{fig:Wass_different_cp}, where similar results are observed. All of these results are consistent with what we would expect from Theorem \ref{thm:Wasserstein_sample_complexity_QMC}, even though the regularity conditions of the theorem are not satisfied in the Gaussian case. The third row of Figure \ref{fig:sample_complexity_unif_gauss} illustrates $|SW_{c,p}(\P^n,\Q^m)|$ with $L=100$ random slices when $c = \|x-y\|_2$ and $p=1$. Clearly, we are able to obtain a gain in accuracy when using RQMC, and this is the case even for large $d$, which is a significant improvement on what is possible with the exact Wasserstein distance. Although the rate is $O(n^{-1}{\color{black}(\log{n})^{{\alpha}_s}})$
regardless of the value of $d$, the gains from using RQMC do become smaller in higher dimensions because the constant in this rate does still depend on $d$.

Finally, the fourth row of Figure \ref{fig:sample_complexity_unif_gauss} looks at the value of $|S_{c,p,\lambda}(\P^n,\Q^m)|$ with $\lambda = 2d$, $p=2$ and $c(x,y)=\|x-y\|_2$. Once again, we observe that RQMC provides significant gains in performance in $d=1$, but also for $d>1$ in the case of the uniform. For the Gaussian, although the performance is improved to some extent for $d>1$, these gains are really small. Interestingly, Figure \ref{fig:Sinkorn_different_cplambda} shows that the gains crucially depend on $p$ and $c$, but also on the choice of $\lambda$. In particular, although Figure \ref{fig:sample_complexity_unif_gauss} could lead us to believe that there are close to no gains for the Gaussian case, this is clearly not the case when using an increased regularisation level $\lambda$.

\begin{figure}[t!]
    \centering
  \includegraphics[width=\textwidth]{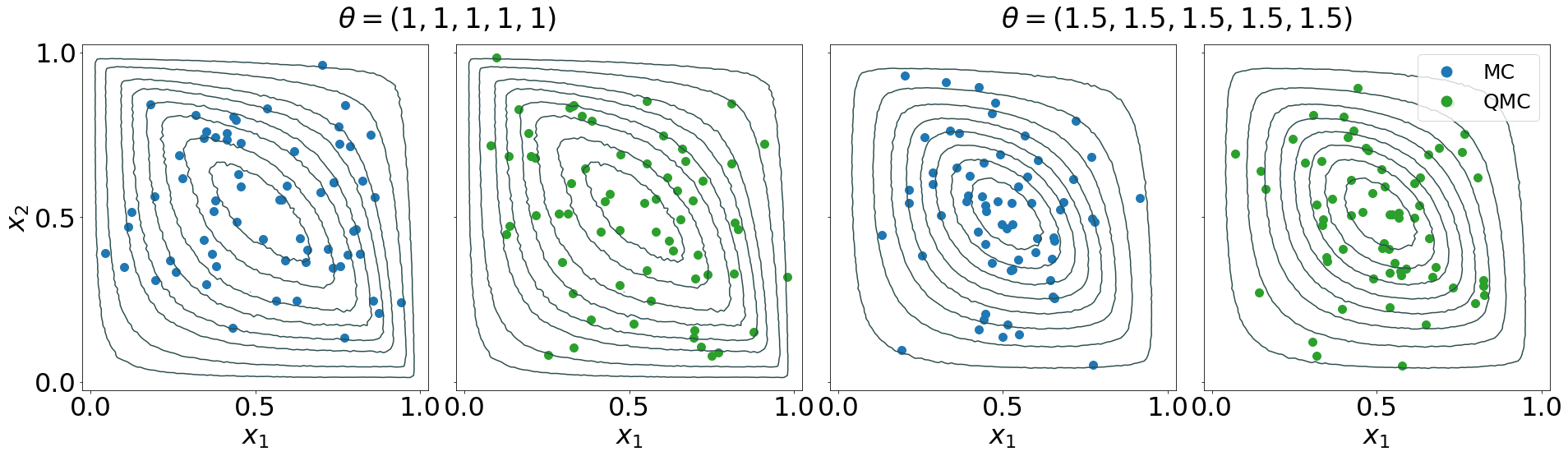}
  \caption{Realisation from the bivariate Beta model using MC and QMC point sets for $\theta = (1,1,1,1,1)$ and $\theta=(1.5,1.5,1.5,1.5,1.5)$. }
  \label{fig:scatter_bivariate_beta}
\end{figure}

\subsection{Inference for Bivariate Beta Distributions}\label{sec:bivariate_beta}

We now move on to studying discrepancy-based inference for intractable generative models with QMC and RQMC. {\color{black} Ever since the work of \cite{Olkin2003}, there has been an interest in designing flexible classes of multivariate distributions which generalise the Beta distribution (as an indicator, \cite{Olkin2003} has over $180$ citations to date). One popular approach is that of \cite{Arnold2011}, which has been used by \cite{Crackel2017} to model household purchasing habits, and by \cite{Sarabia2014} to model indicators of well-being. Although flexible, this does lead to an intractable density which makes inference challenging}. We will focus on the $d=2$ and $p=5$ version of the model previously considered by \cite{Jiang2018,Nguyen2020b}, and whose marginals are $\text{Beta}(\theta_1+\theta_3,\theta_4+\theta_5)$ and $\text{Beta}(\theta_2+\theta_4,\theta_3+\theta_5)$ distributed respectively in the first and second coordinate. In particular, denoting by $\lf x \rf$ the integer part of some $x \in \R$:
\begin{talign*}
G_\theta^1(u) & := \frac{\tilde{u}_1+\tilde{u}_3}{\tilde{u}_1+\tilde{u}_3+\tilde{u}_4+\tilde{u}_5}, 
\quad 
G_\theta^2(u) := \frac{\tilde{u}_2+\tilde{u}_4}{\tilde{u}_2+\tilde{u}_3+\tilde{u}_4+\tilde{u}_5}, 
\quad
\tilde{u}_i  = - \sum_{k=1}^{\lf \theta_i \rf} \ln (u_{ik}) + u_{i0},
\end{talign*}
where $u_{i0} \sim \text{Gamma}(\theta_i - \lf \theta_i \rf, 1)$, $u = (u_{11}, \ldots, u_{1 \theta_1},u_{21},\ldots,u_{5 \theta_5}) \sim \text{Unif}([0,1]^{s})$ and $s = \sum_{i=1}^5 \lf \theta_i \rf$. Note that the dimension $s$ of the base space now depends on the value of $\theta$.

In the special case where $\theta_i$ is an integer, $u_{i0}=0$ is fixed (as opposed to sampled from a Gamma). In this case, both $\X$ and $G_\theta$ satisfy the conditions in Assumptions \ref{assumptions:domain} and \ref{assumptions:generator}. When this is not the case, $u_{i0}$ can be generated through rejection sampling (see Appendix \ref{appendix:bivariate_beta}). In that case, the generator does not satisfy Assumption \ref{assumptions:generator} anymore, and also has a much higher-dimensional domain; i.e. $s = \sum_{i=1}^5 \lf \theta_i \rf + 15$. Here, the first term comes from the simulation of Gamma random variables with integer parameters $\lf \theta_1 \rf, \ldots, \lf \theta_5 \rf$, and the second term is the dimensionality required to simulate five Gamma random variables with scalar parameters in $(0,1)$ (that is, the simulation of a Gamma through rejection sampling requires a three-dimensional point). Despite these challenges, we will see below that certain gains in performance are still possible.

\begin{figure}[t!]
    \centering
  \includegraphics[width=0.75\textwidth]{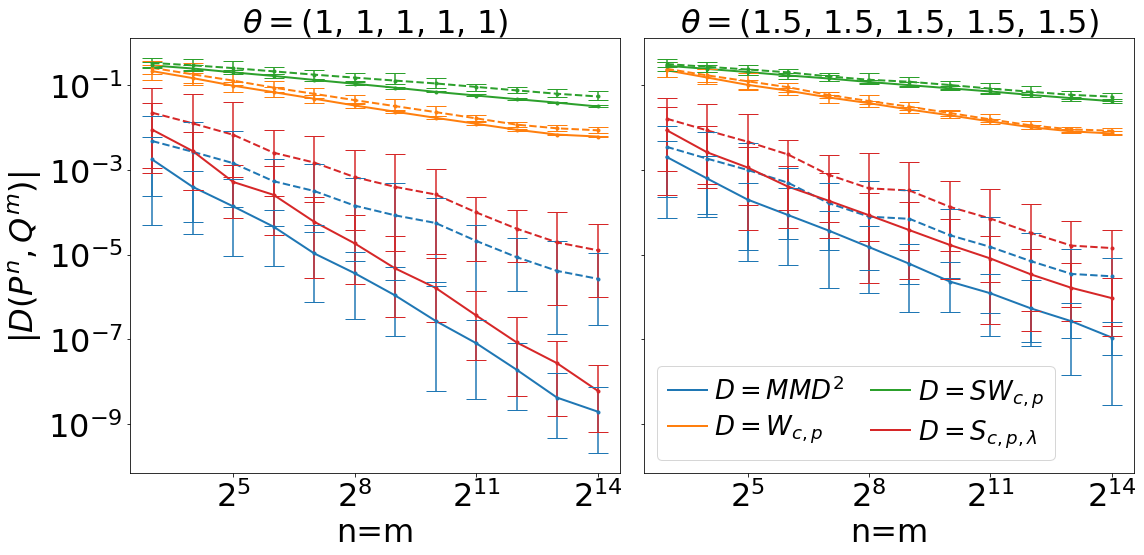}
  \caption{Sample complexity results for the bivariate Beta distribution with integer and scalar parameters. In both cases d=2, but the left-hand side plot uses $s=5$ whereas the right-hand side plot uses $s=20$. These differences in dimensionality seem to impact the convergence rate obtained through QMC point sets. Each solid line corresponds to RQMC, while the dashed lines correspond to MC point sets.}
  \label{fig:samplecomplexity_bivariate_beta}
\end{figure}

Figure \ref{fig:scatter_bivariate_beta} provides realisations from this model through MC and QMC sampling (in blue and green respectively). As observed, the QMC point set provides a slightly better coverage of the distribution, although the difference is not very large visually. In those cases, $s=5$ for the left-hand side plot, whereas $s = 20$ for the right-hand side plot. We note that this case significantly differs from the examples in the previous section since we have $s \gg d$, which may partly explain why the difference is not as striking visually. However, looking at Figure \ref{fig:samplecomplexity_bivariate_beta} (which is the equivalent of Figure \ref{fig:sample_complexity_unif_gauss} for this model), we can see that QMC leads to a significant improvement in terms of sample complexity, especially in the case of integer parameter and to a lesser extent with scalar parameter values. Once again, this difference between left-hand side and right-hand side plot is most likely due the difference in value of $s$, and the fact that Assumption \ref{assumptions:generator} is not satisfied in the latter case. We also note that the advantage provided by QMC is particularly significant for the MMD and the sliced Wasserstein distance.

\begin{figure}[t!]
    \centering
  \includegraphics[width=\textwidth]{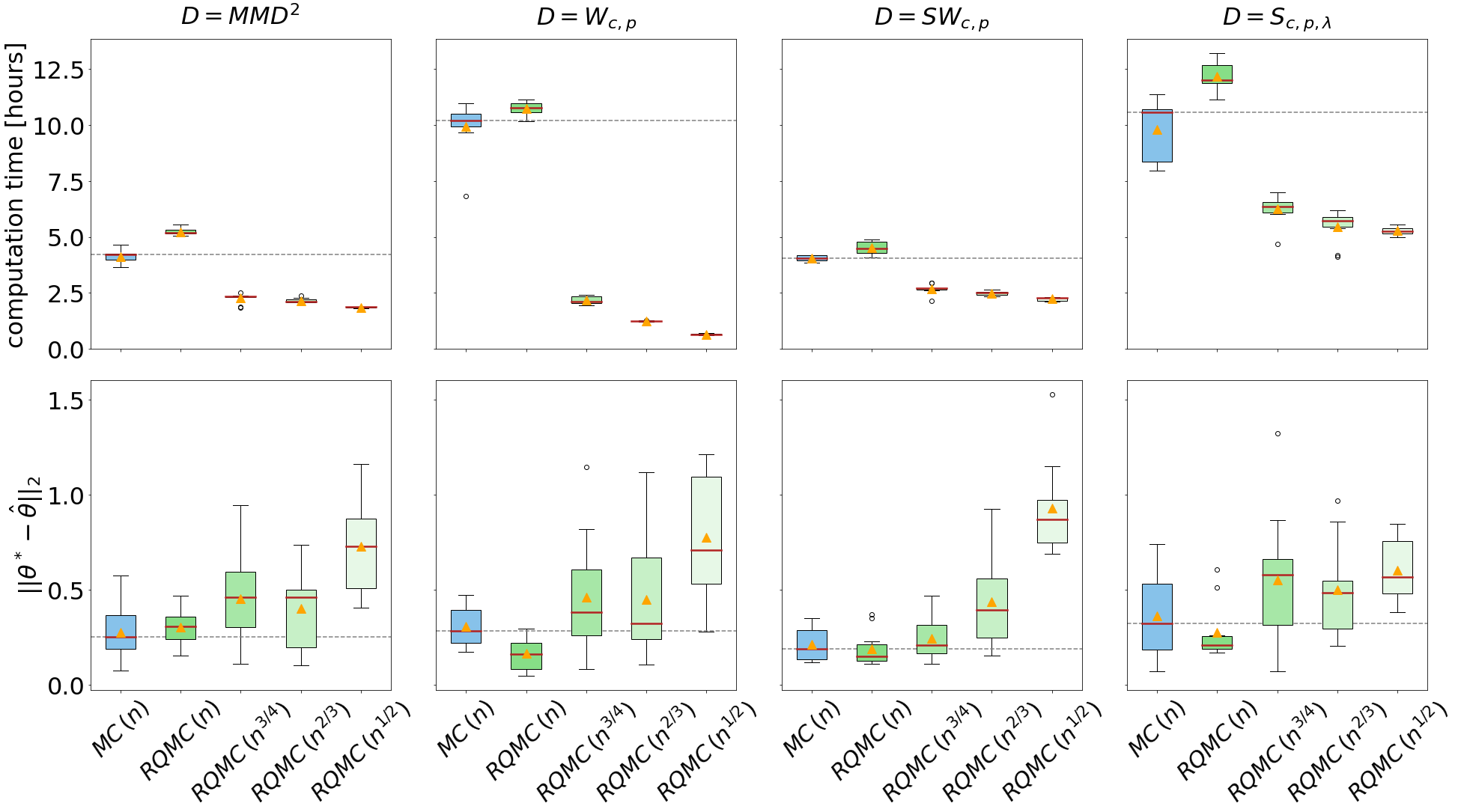}
  \caption{Minimum distance estimation for the parameters of the bivariate Beta distribution with the MMD, Wasserstein, sliced Wasserstein and Sinkhorn divergence. Each box-plot correspond to the result of $10$ repetitions of an identical experiment, and the red bars provide the sample median. The dotted horizontal line correspond to the median value for the MC-based estimators with $n$ points.}
  \label{fig:MDE_bivariate_beta}
\end{figure}

For the last part of this experiment, we perform inference for the parameter $\theta$ using an MDE approach with the MMD, the Wasserstein distance, the sliced Wasserstein distance and the Sinkhorn divergence. The generator $G_\theta$ is not differentiable in $\theta$ and we therefore propose to use a gradient-free global optimisation algorithm. We utilise the differential evolution algorithm due to \cite{Storn1997}, which is implemented as a sub-routine of the \texttt{optimize} function in the python library \texttt{SciPy} \cite{SciPy2020}. The dataset consists of $m=2^{16}$ points, from which a minibatch of $2^{10}$ points is sampled at random at every iteration. Depending on the considered experiment, either $n$ samples are generated using MC or $n, n^{\frac{3}{4}}, n^{\frac{2}{3}}$ or $n^{\frac{1}{2}}$ are simulated using RQMC at every iteration. The optimisation algorithm is run for 3,000 iterations for every setting. For the MMD, a squared-exponential kernel with lengthscale $l=1.5d^{1/2}$ is used. The Wasserstein distance is computed with $c(x,y)=\|x-y\|_2$ and $p=1$ as is the sliced Wasserstein distance based on 100 projections. The Sinkhorn divergence is considered with $\lambda=5d$, $c(x,y)=\|x-y\|_2$ and $p=2$. For the experiments, we focused on the case where $\theta^* = (\theta_1^*,\theta_2^*,\theta_3^*,\theta_4^*,\theta_5^*) =(1,1,1,1,1)$ as this was studied by \cite{Jiang2018,Nguyen2020b}. Therefore, the bounds, within each parameter is optimised by the differential evolution algorithm, are set to $[0,2]$.  We note that although the true parameter is integer valued, the optimisation algorithm will have to simulate data for parameter values which are scalar-valued. As a result, the dimensionality of the domain of the generator will generally be $s \in [15,25]$ (assuming that the optimisation routine does not explore regions of the parameter space with large parameter values relative to $\theta^*$). 

The results of our experiments are presented in Figure \ref{fig:MDE_bivariate_beta}, where we studied the computational cost and the accuracy of the estimates in $l_2$ norm for each choice of discrepancy. In each of these settings, we compared an MC method based on $n$ points with an RQMC method with $n, n^{\frac{3}{4}}, n^{\frac{2}{3}}$ and $n^{\frac{1}{2}}$ points. As could be reasonably expected, the RQMC-based estimator with $n$ points is significantly more expensive than an MC with $n$ points, but it is also much more accurate in $l_2$ error. Similarly, the RQMC-based estimators with $n^{\frac{2}{3}}$ or $n^{\frac{1}{2}}$ are less accurate in $l_2$ error but usually cheaper than MC with $n$ points.
{\color{black}
More interestingly, we see that the RQMC estimator with $n^{\frac{3}{4}}$ points is both cheaper and more accurate than the MC estimator with $n$ points for the Wasserstein distance, whilst for the Sinkhorn divergence it is cheaper and provides roughly the same level of accuracy.}
This clearly highlights that RQMC point sets can provide advantages even in cases not necessarily covered by our theoretical results. Surprisingly, this is not the case for the MMD, for which the performance of the RQMC-based estimator with $n^{\frac{3}{4}}$ is slightly worse than for the MC-based estimator in this experiment. We speculate that this may be due to a poor choice of kernel or an issue with the optimisation method since we obtained encouraging sample complexity results in Figure \ref{fig:samplecomplexity_bivariate_beta}. 


\begin{figure}[t!]
    \centering
  \includegraphics[width=0.7\textwidth]{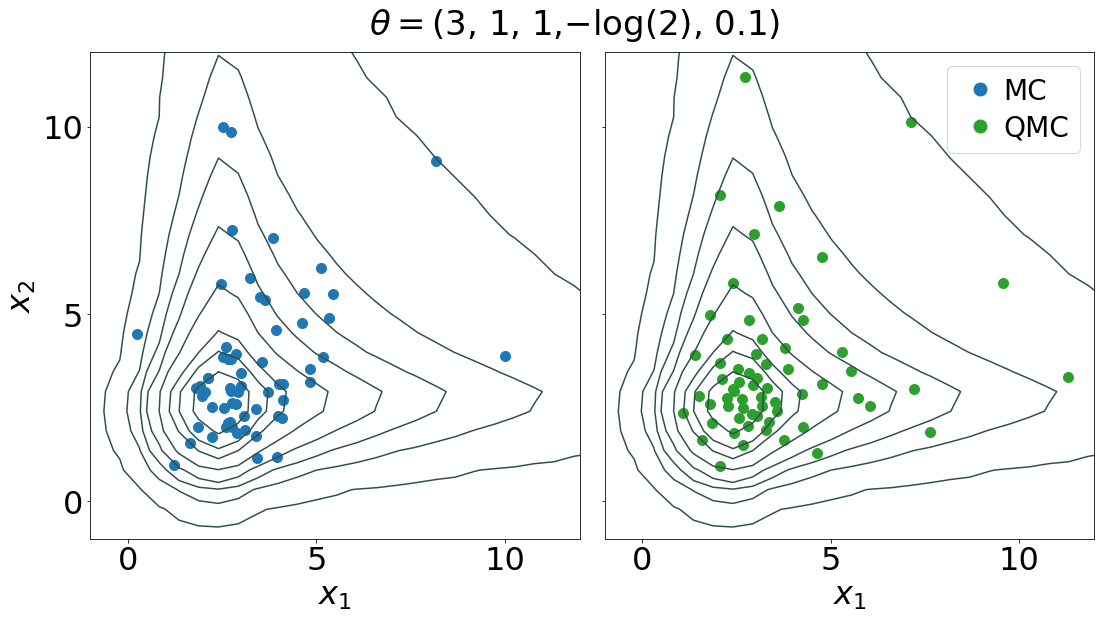}
  \caption{Realisations from the multivariate g-and-k distribution for $d=2$. The scatter plots for MC and QMC are both based on $n=2^6$ realisations, of which two fall outside the plotted interval for each point set. These points were omitted to allow for a zoomed-in view of the mode of the distribution.}
  \label{fig:gandk_d2_scatter}
\end{figure}

\subsection{Inference for Multivariate g-and-k Models}

\color{black}{Next}, we will consider is the multivariate extension of the g-and-k distribution considered in \cite{Jiang2018,Nguyen2020b}. This parametric class is very flexible as it contains four parameters controlling the mean, variance, skewness and kurtosis of the marginals, as well as a fifth parameter controlling correlations across neighbouring coordinates. Unfortunately, inference is made challenging by the fact that the density is not available in closed-form. It is however straightforward to sample from this distribution, and it has recently become one of the most common target problems to assess the performance of inference schemes for generative models; see e.g. \cite{Prangle2017,Bernton2017,Jiang2018,Bernton2019,Briol2019,Nguyen2020b,Dellaporta2022} for a small subset of recent papers using this model. {\color{black} The g-and-k has been applied to a range of applied problems, including (amongst others) insurance modelling \cite{Peters2016}, ranking and selection \cite{Haynes1997}, and modelling of the prices of short-term rentals \cite{Rodrigues2020}. }

The generator for this model is:
  \begin{talign*}
G_\theta(u)  & := \theta_1 + \theta_2 \left(1 + 0.8 \frac{(1 - \exp(-\theta_3 z)) }{(1 + \exp(-\theta_3 z) )}\right) (1 + z^2)^{\theta_4} z, 
\end{talign*}
where $\U = \text{Unif}([0,1]^d)$, $z = \Sigma^{\frac{1}{2}}\Phi^{-1}(u)^\top$ where $\Phi^{-1}(u)$ is the inverse CDF of the univariate standard Gaussian distribution applied element-wise and $\Sigma \in \R^{d\times d}$ is a symmetric tri-diagonal Toepliz matrix with diagonal entries all equal to $1$ and off-diagonal entries equal to $\theta_5$. Its square-root can be obtained in closed form and is provided in Appendix \ref{appendix:gandk}. Note that $s=d$ and we can straightforwardly replace MC realisations with a QMC or RQMC point set. Another important remark is that this generator does not satisfy the conditions of Assumption \ref{assumptions:generator} since we are using the inverse CDF of a standard Gaussian. \color{black}{As parameter of interest, we consider $\theta=(\theta_1,\theta_2,\theta_3,\exp(\theta_4),\theta_5)$, where the rescaling of is used to avoid numerical instabilities during optimisation.}

\begin{figure}[t!]
    \centering
  \includegraphics[width=\textwidth]{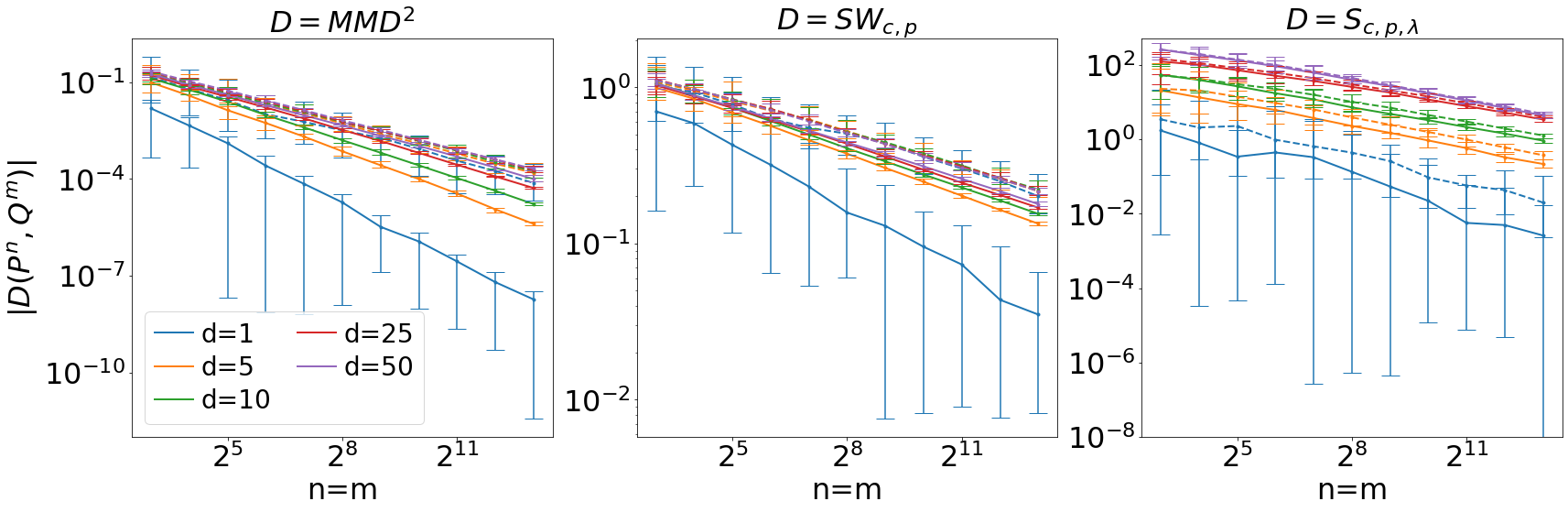}
  \caption{Sample complexity results for the multivariate g-and-k distribution in various dimensions $d$. The Wasserstein distance is omitted due to the fact that the discrepancy is prohibitively expensive when $d>1$ and $n$ is large. Each solid line corresponds to RQMC, whereas the dashed lines correspond to MC point sets.}
  \label{fig:gandk_d2_samplecomplexity}
\end{figure}

Figure \ref{fig:gandk_d2_scatter} presents a scatter plot of two point sets of size $n=2^6$ obtained through MC and RQMC in the case where $\theta = (3,1,1,-\log(2),0.1)$. We can observe that the RQMC-based point set provides a better coverage of areas of high probability than the MC-based point set. These visual results also bare out in the estimates of the discrepancies in Figure \ref{fig:gandk_d2_samplecomplexity}, where we plot the sample complexity as a function of $n$ for different values of $d$ for the MMD (with squared-exponential kernel and lengthscale $l=1.5d^{1/2}$), the sliced Wasserstein distance (with $c(x,y)=\|x-y\|_2$, $p=1$ and 100 projections) and the Sinkhorn divergence (with $\lambda=5d$, $c(x,y)=\|x-y\|_2$ and $p=2$). Here, the Wasserstein distance is omitted due to the prohibitive computational cost when $d$ is large.

\begin{figure}[t!]
    \centering
    \includegraphics[width=0.9\textwidth]{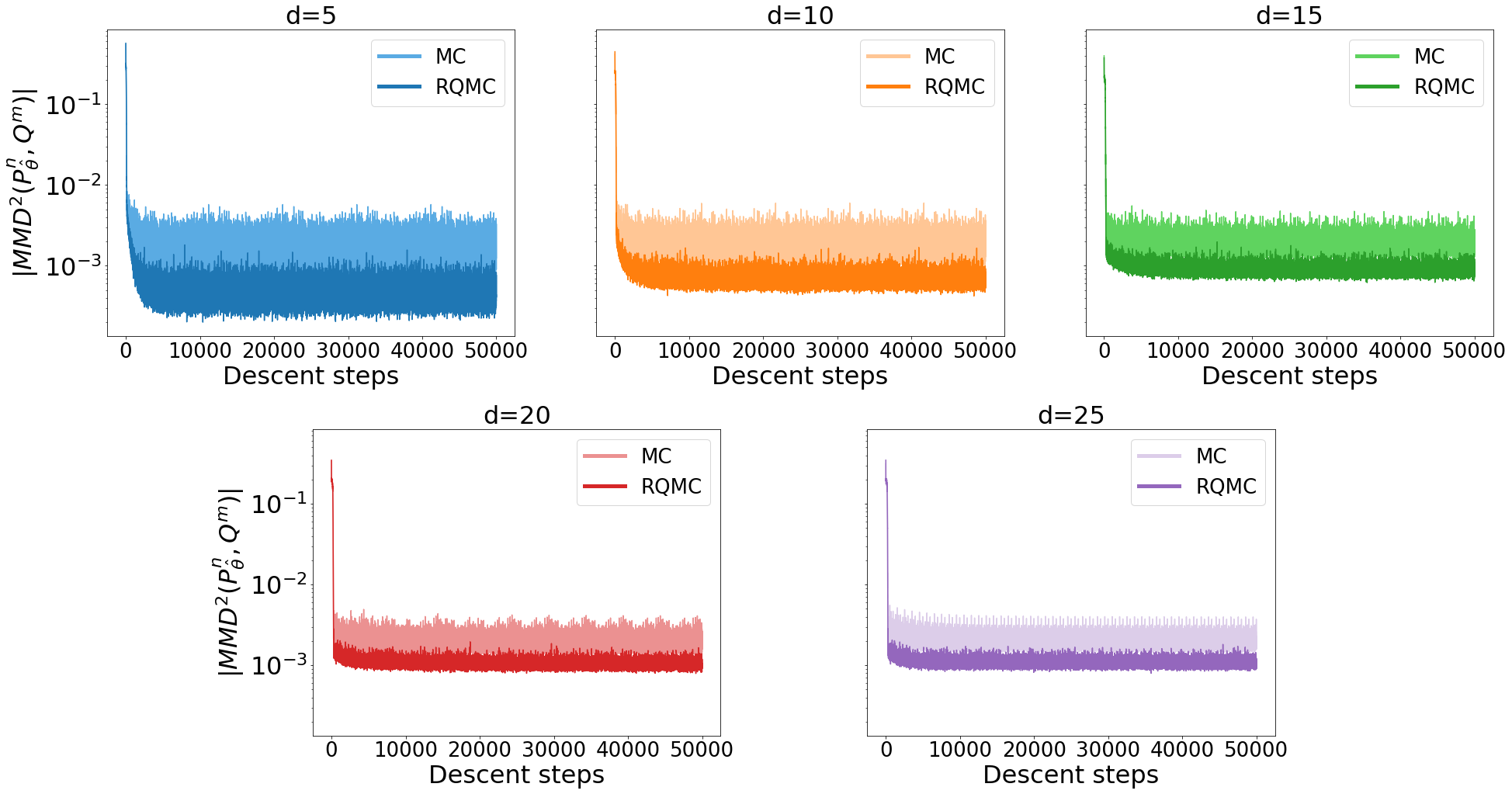}
  \caption{Minimum distance estimation with the MMD for the multivariate g-and-k distribution. The figure plots the estimated MMD between the model with the estimated parameter and the data as a function of the number of stochastic gradient descent steps.}
  \label{fig:gandk_d2_optim_MMD}
\end{figure}

In each case, the RQMC algorithms significantly outperform their MC counterpart, although this improved performance is limited for higher values of $d$. For example, in the case of MMD, the RQMC rates were of the form $n^{-\alpha}$ with $\alpha$ equal to $0.73, 0.65, 0.58$ and $0.54$ in dimensions $5, 10, 25$ and $50$ respectively, whereas $\alpha$ was approximately $0.5$ for MC in all cases. This is in line with what we would expect following the results of Section \ref{sec:experiments_unif_gauss} where the use of the inverse CDF of a Gaussian was studied in detail.

In the last part of our experiments for the multivariate g-and-k distribution, we adapt a gradient-based optimisation method to perform inference for the parameter $\theta$ using an MDE approach that builds on the MMD. The considered stochastic gradient descent (SGD) algorithm is similar to the one of \cite{Briol2019}, but uses an approximation of the squared MMD using empirical measures as in (\ref{eq:MMD_Vstat_defn}) instead of a U-statistic approximation. From $m=2^{16}$ data points, a minibatch of $2^{11}$ points is sampled for every descent step. Using either the MC or QMC approach, $n=2^{9}$ data points are simulated for each descent step. The step size of the SGD algorithm is fixed at $0.2$ (for both MC and QMC) and the optimisation is run for $50,000$ descent steps. The squared MMD and its gradient are computed based on the squared-exponential kernel with lengthscale $l=1.5d^{1/2}$.To obtain the gradient of the multivariate g-and-k distribution, we make use of automatic differentiation provided by the python library \texttt{JAX} \cite{Jax2018}. The experiments aim at retrieving the true parameter $\theta^*=(\theta_1^*,\theta_2^*,\theta_3^*,\exp(\theta_4^*),\theta_5^*)=(3,1,1,-\log(2),0.1)$ and start the SGD algorithm at $\theta^0= (0.3,0.3,0.3,0.3,0.3)$.

\begin{figure}[t!]
    \centering
      \includegraphics[width=\textwidth]{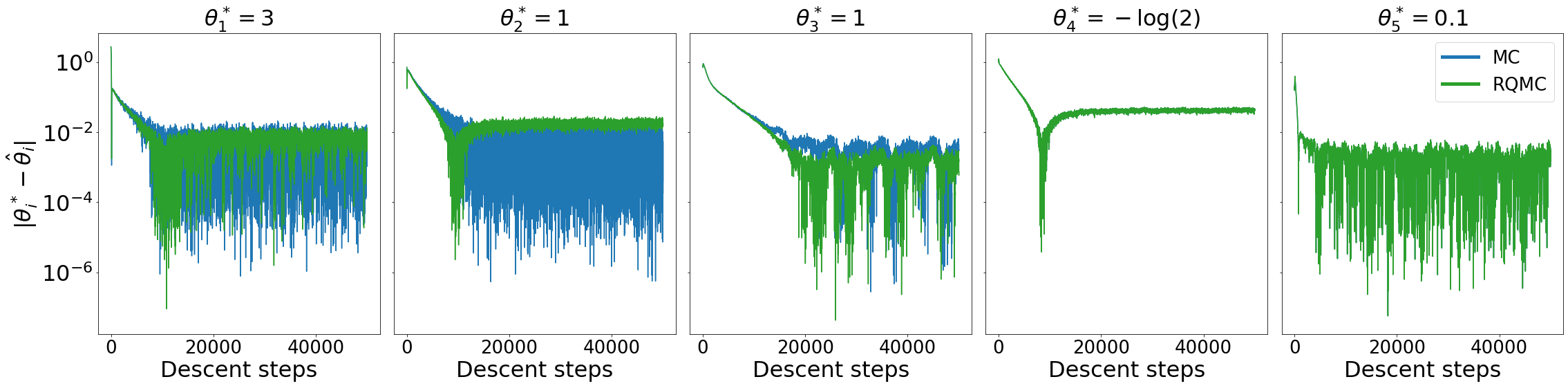}
  \caption{Minimum distance estimation with the MMD for the multivariate g-and-k distribution. The figure plots the $l_1$ error between the estimated parameter and the true value as a function of the number of stochastic gradient descent steps.}
  \label{fig:gandk_d2_optim_d5}
\end{figure}

Figure  \ref{fig:gandk_d2_optim_MMD} illustrates the results for $|\mMMD^2(\P^n_{\hat{\theta}},\Q^m)|$ as a function of the number of descent steps of the stochastic gradient descent method, where $\Q^m$ here corresponds to the entire original data (i.e. $m=2^{16}$). The experiment is repeated for a range of values of $d$ between $5$ and $25$. As we can observe, the estimated MMD is much more accurate when sampling is done with RQMC. In Figure \ref{fig:gandk_d2_optim_d5}, we then look at the case of $d=5$ in more details. In particular, the figure shows the $l_1$ distance between the true and estimated parameters of the g-and-k as a function of the number of descent steps. \color{black}{It is overall unclear which of RQMC and MC outperforms the other, and this depends on which parameters are of most interest. RQMC seems to outperform MC for $\theta_3$, performs equally well as MC for $\theta_4$ and $\theta_5$ (the curves overlap), and tends to do worse for $\theta_2$. For all parameters, the jumps in $l_1$ error between descent steps is much larger for MC than RQMC, highlighting that RQMC estimates have a much smaller variance. The contrast between Figure \ref{fig:gandk_d2_optim_MMD} and Figure \ref{fig:gandk_d2_optim_d5} highlights that minimisation of a discrepancy does not necessarily mean that the estimates for all parameter values will be accurate.} In fact, 
Figure \ref{fig:gandk_d2_optim} in Appendix \ref{appendix:gandk} actually shows that RQMC actually has a worse performance than MC as $d$ grows when looking at the results in terms of $l_2$ errors instead of MMD (as in Figure \ref{fig:gandk_d2_optim_MMD}). In this case, we expect that such a counter-intuitive result is due to the value of $m$ being too small relative to that of $n$ for RQMC, which could lead to over-fitting. 

{
\color{black}
\subsection{Inference for Generative Neural Networks}

Our final model is a generative neural network which was trained using the Sinkhorn divergence by \cite{Genevay2019}. More precisely, this model is the decoder network of a variational autoencoder (VAE) given by $G_\theta:\mathcal{U} \rightarrow \X$ with $\mathcal{U} = [0,1]^2$ and $\X = [0,1]^{784}$ (i.e. $s=2$ and $d = 784$) where:
\begin{align*}
    G_\theta(u) = \phi_2(\phi_1(\phi_1(u^\top W^1 +b^1)^\top W^2+b^2)^\top W^3+b^3)
\end{align*}
and $\theta$ is a vector containing all entries of the weight matrices $W^1 \in \R^{2 \times 500}, W^2 \in \R^{500 \times 500}, W^3 \in \R^{500 \times 784}$ and biases $b^1 \in \R^{500}, b^2 \in \R^{500}, b^3 \in \R^{784}$ so that $p = 644784$. Additionally, $\phi_1(x) = \log(\ \exp(x)+1)$ (a softplus activation function) and $\phi_2(x) = (1+\exp(-x))^{-1}$ (a logistic activation function), and the output of the generator is a $784-$dimensional vector which can be rescaled to form a $28 \times 28$ pixel image. Since $G_\theta$ is the composition of smooth functions, it is itself smooth. Furthermore, since $\X$ is bounded, the derivatives of $G_\theta$ (which are continuous) must also be bounded, and $G_\theta$ therefore satisfies Assumption \ref{assumptions:generator}, and hence our theorems hold.

\begin{figure}[t!]
    \centering
    \includegraphics[width=0.48\textwidth]{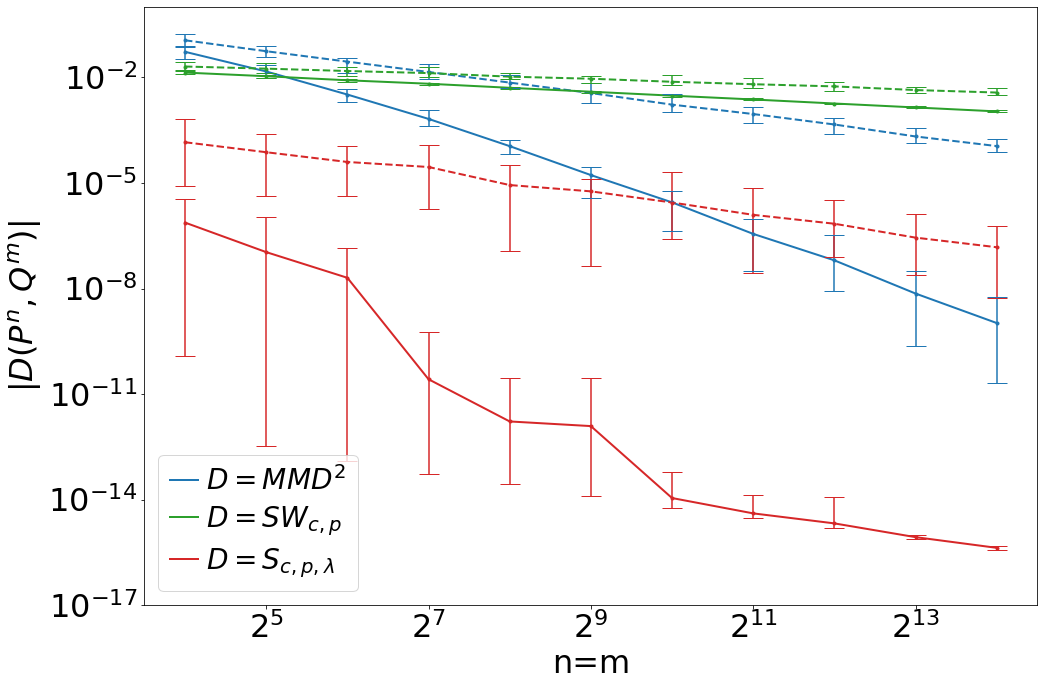}
    \includegraphics[width=0.48\textwidth]{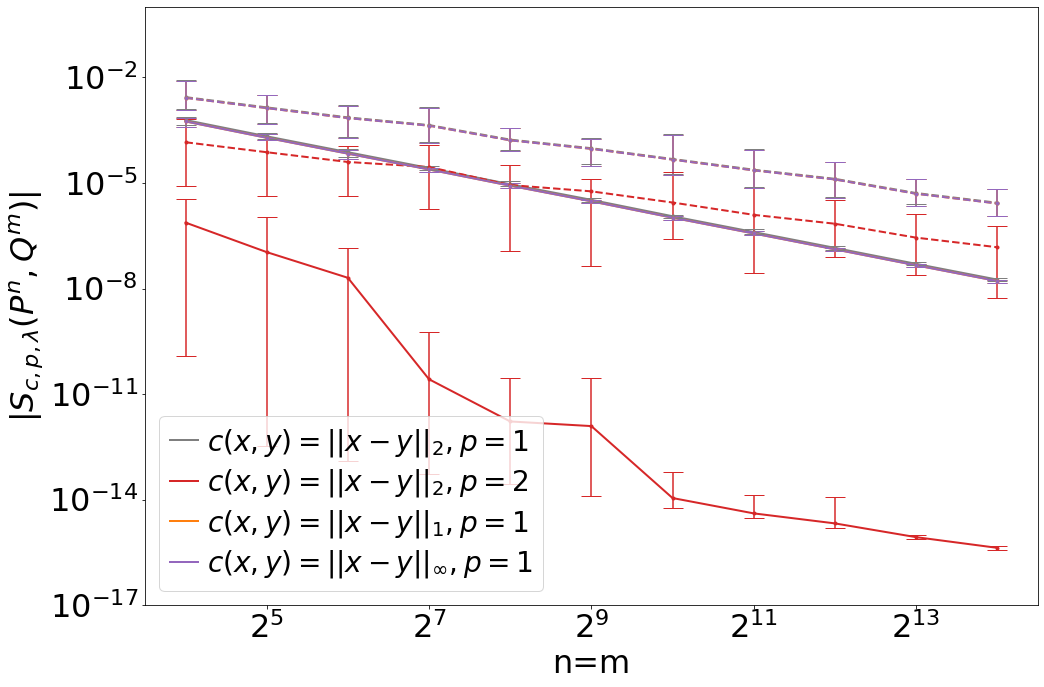}
    \caption{\color{black}{Sample complexity results for the generative neural network. Each solid line corresponds to RQMC, whereas the dashed lines correspond to MC.}
    }
    \label{fig:vae_sample_complexity}
\end{figure}

In the right-hand side plot of Figure \ref{fig:vae_sample_complexity}, the sample complexity is plotted as a function of $n$ for the MMD (with squared-exponential kernel and lengthscale $l=0.01$), the sliced Wasserstein distance (with $c(x,y)=\|x-y\|_2$, $p=1$ and 100 projections), and the Sinkhorn divergence (with $c(x,y)=\|x-y\|_2$, $p=2$ and $\lambda=1$). Here, the Wasserstein distance is omitted due to the prohibitive computational cost in high dimensions. We observe that QMC leads to significant improvements in sample complexity, especially for the MMD and Sinkhorn divergence.

Comparing the sample complexity for the Sinkhorn divergence with different choices of cost $c$ and order $p$ in the right-hand side plot of Figure \ref{fig:vae_sample_complexity}, we find that the choice of squared Euclidean cost, i.e. $c(x,y)=\|x-y\|_2$ and $p=2$, significantly outperforms the other considered choices. It is therefore used in the following experiments.

\begin{figure}[t!]
    \centering
    \includegraphics[width=0.48\textwidth]{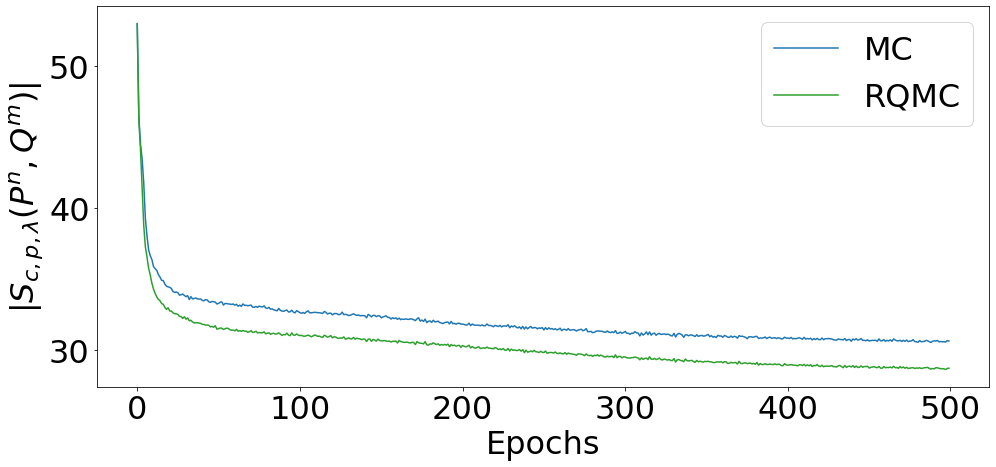}
    \includegraphics[width=0.48\textwidth]{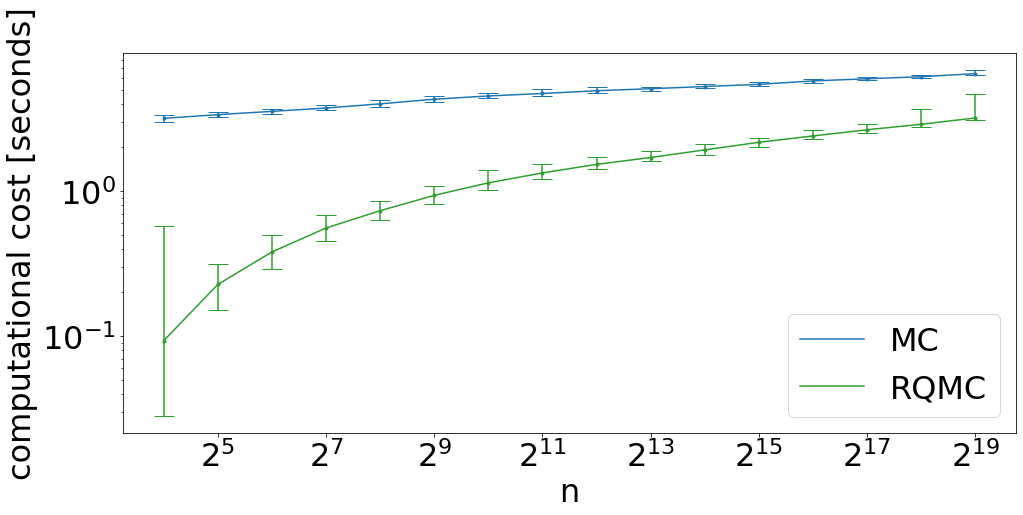}
    \caption{\color{black}{Left: Minimum distance estimation with the Sinkhorn divergence for the VAE. The figure plots the Sinkhorn divergence between the model with the estimated parameters and the data as a function of the number of training epochs. Right: Comparison of the computational cost for simulating from the generative neural network using MC and RQMC. Each line represents the average of 500 repetitions and the error bars give the minimum and maximum values observed.} 
    }
    \label{fig:vae_optim_loss}
\end{figure}

In the final experiment, the generative neural network is trained as the decoder network of a VAE on the MNIST dataset for 500 epochs with mini-batches of size 300 using the Adam optimizer \cite{Kingma2014Adam}. This MDE approach is based on the Sinkhorn divergence with parameters $\lambda=1$, $c(x,y)=\|x-y\|_2$ and $p=2$ and a dataset of size $m=n=55,000$. This setup corresponds to the one used by \cite{Genevay2018}. The implementation of this experiment uses the python library \texttt{TensorFlow} \cite{Tensorflow2015} and \texttt{SciPy}\cite{SciPy2020} to generate Sobol points. Using RQMC sampling, we observe in the left-hand side plot of Figure \ref{fig:vae_optim_loss} that the training loss decreases significantly faster in the number of training epochs than when using MC sampling.

The right-hand side plot of Figure \ref{fig:vae_optim_loss} compares the computational cost of simulating from the generative neural network, which implies sampling with $d=2$. We observe that RQMC sampling is much cheaper than MC for all considered $n$.
}


\section{Conclusion} \label{sec:conclusions}

This paper focused on the use of QMC and RQMC point sets for discrepancy-based inference in intractable generative models. We showed (in Theorems \ref{thm:QMC_concentration_ineq}, \ref{thm:Wasserstein_sample_complexity_QMC} and  \ref{thm:Sinkhorn_QMC}) that the sample complexity becomes $O(n^{-1}{\color{black}(\log{n})^{{\alpha}_s}})$ instead of $O(n^{-1/2})$ for the MMD and Sinkhorn divergence in arbitrary dimension $d$. These faster rates can provide significant improvements on the current state-of-the-art with no significant increases in computational cost (since QMC point sets can be pre-computed). Unfortunately, the rate for the Wasserstein-1 distance can only be improved when $d=1$, and is otherwise gated at $O(n^{-1/d})$ due to a well-known curse of dimensionality. However, we showed that the recently introduced sliced-Wasserstein distance can obtain the optimal $O(n^{-1}{\color{black}(\log{n})^{{\alpha}_s}})$ rate regardless of $d$.

One significant drawback of our results is that they not only require the generator to satisfy certain regularity conditions (see Assumption \ref{assumptions:generator}), but also that $\X$ is compact (see Assumption \ref{assumptions:domain}). These are common assumptions for the QMC literature (see the discussion in \cite{Owen2013}), but these nonetheless exclude many cases of practical interest. Despite these limitations, we showed in Section \ref{sec:experiments} that QMC/RQMC can still provide significant gains when the assumptions do not hold; for example when using the inverse transform approach to sampling from Gaussian distributions (which has an unbounded generator) and sampling Gamma random variables through rejection sampling. This is in line with work in the QMC literature (see for example \cite{Owen2006}) and future work could explore these cases from a theoretical viewpoint in more detail.

Another potential line of future research would be to explore the use of other point sets, including weighted point sets, for inference in generative models. This was recently studied in the context of the Sinkhorn divergence by \cite{Beugnot2021}, who use quantization to improve sample qualities. However, alternative approaches could also be used. For example, Bayesian quadrature \cite{Briol2019b} is known to provide optimally weighted point set for the MMD, and could lead to faster sample complexity results. We expect that such approaches could provide significant improvements in performance, particularly in cases of computationally expensive generators. Higher-order digital nets could also be used to provide dimension-independent convergence rates, albeit with further assumptions on the generator. \color{black}{In this respect, one could think of adapting the architecture of deep generative models as well as the choice of discrepancy so as to ensure that such fast rates can be obtained.}

\subsubsection*{Acknowledgments}

The authors are grateful to Chris Oates {\color{black} and two anonymous reviewers} for helpful comments and suggestions on this paper. FXB was supported by the Lloyd’s Register Foundation programme on data-centric engineering at The Alan Turing Institute under the EPSRC grant [EP/N510129/1].

{\scriptsize
\printbibliography
}

\newpage
\appendix
{
\begin{center}
\LARGE
    \vspace{5mm}
    \textbf{Appendix}
    \vspace{5mm}
\end{center}
}

First, in Appendix \ref{appendix:background}, we recall relevant background material on QMC. Then, in Appendix \ref{appendix:proofs}, we provide all the proofs for the results in the main text. 
In Appendix \ref{appendix:numerical_experiments}, we provide additional numerical experiments to complement the results in the main text.


\section{Additional Background} \label{appendix:background}

For completeness, we first recall several definitions and results which are relevant for QMC. Our presentation closely follows \cite{Owen2005}, and we refer the reader to this paper for further details. For some vector $u \in \mathbb{R}^s$, we will denote its $j$'th component as $u_j$, so that $u=(u_1,\ldots,u_s)$.
 We first introduce the $s-$fold alternating sum of $f$ over $[a,b] \subset \R$: $$\Delta(f;a,b)=\sum_{v\subseteq\{1,\ldots,s\}}(-1)^{\left|v\right|}f(a_{v}:b_{-v})$$ 
 where $a,b$ are two $s$ dimensional vectors. We write $|v|$ for the cardinality of the multi-index $v$, and $-v$ for the sequence $\{1,\ldots,s\} \backslash v$ which contains all elements of $\{1,\ldots,s\}$ not in $v$. Furthermore, $a_{v}$ denotes a $\left|v\right|$-tuple of real values representing the components $a_{j}$ for $j\in v$. The symbol $a_{v}:b_{-v}$ represents the point $u\in[a,b]^{s}$ with $u_{j}=a_{j}$ for $j\in v$, and $u_{j}=b_{j}$ for $j\notin v$. 
 
 Let $\mathcal{Y} = \{ u \in [a,b]^s \; | \; 0 < u_1 < u_2< \ldots < u_s =1\}$ be a ladder on $[a,b]$. For $j=1,2,\ldots,s$, denote by $\mathcal{Y}^{j}$ a ladder on $[a_{j},b_{j}]$. A (multi-dimensional) ladder on $[a,b]$ has the form $\mathcal{Y}=\prod_{j=1}^{s}\mathcal{Y}^{j}$. For $y\in\mathcal{Y}$, the successor point $y_{+}$ is defined by taking $(y_{+})_j$ to be the successor of $y_{j}$ in $\mathcal{Y}^{j}$. The variation of $f$ over $\mathcal{Y}$ is then given by:
 \begin{talign*}
 V_{\mathcal{Y}}(f)=\sum_{y\in\mathcal{Y}}\left|\Delta(f;y,y_{+})\right|.
 \end{talign*}
 Let $\mathbb{Y}^{j}$ denote the set of all ladders on $[a_{j},b_{j}]$ and put $\mathbb{Y}=\prod_{j=1}^{s}\mathbb{Y}^{j}$. 
Then, the \emph{variation of $f$ in the sense of Hardy and Krause} is given by:
\begin{talign*}
    \VHK(f)=\sum_{v\subsetneq \{1,\ldots,s\}}V_{[a_{-v},b_{-v}]}f(u_{-v}:b_{v}).
\end{talign*}
We can now finally present the Koksma-Hwlaka inequality, which decouples the quadrature error into a term depending on the function, the Hardy-Krause variation, and a term depending on the point set, the star discrepancy. 
\begin{lemma}[Theorem 15.5 in \cite{Owen2013}] \label{lemma:KoksmaHwlaka_BVHK}
  Let $\mathcal{U}=[0,1]^s$, $f: \mathcal{U}\rightarrow \R$ and $\{u_{i}\}_{i=1}^{n} \subset  \mathcal{U}$. Then, if $V_{\text{HK}}(f)<\infty$, we have:
  \begin{talign*}
    \left| \int_{[0,1]^s} f(u) du-\frac{1}{n}{\sum_{i=1}^{n}}f(u_i)\right|
    & \leq \VHK(f) D^{*}(\{u_i\}_{i=1}^n).
  \end{talign*}
\end{lemma}
Combining this result with the definition of QMC or RQMC point set allows us to provide results on the convergence of QMC/RQMC estimators for functions with bounded Hardy-Krause variation.


\section{Proof of Theoretical Results}\label{appendix:proofs}

In this appendix, we provide proofs of all the theoretical results in the main text. 
Firstly, in Section \ref{appendix:proof_preliminary}, we provide some useful preliminary results. Then, Section \ref{appendix:proof_MMD} contains the proof of our results on MMD, Section \ref{appendix:proof_Wasserstein} the proof of our results on the Wasserstein distance, and Section \ref{appendix:proof_sinkhorn} the proof of our results on the Sinkhorn divergence.

\subsection{Preliminary Results}\label{appendix:proof_preliminary}

Before stating our main result for this section, Theorem \ref{thm:general_Sbolev_composition}, we recall a preliminary results which will be used in its proof.  
\begin{lemma}[Generalised H\"older's Inequality; Corollary 2.6 of \cite{Adams2006}]\label{lemma:Holder}
Suppose that $p, p_1, \ldots, p_r \in (0,\infty]$ and $\sum_{i=1}^r p_{i}^{-1} = p^{-1}$. Then, if $\|f_i\|_{L^{p_i}(\X)} < \infty$ for all $i \in \{1,\ldots,r\}$, we have:
\begin{talign*}
\| \prod_{i=1}^r f_i \|_{L^r(\X)} \leq \prod_{i=1}^r \|f_i\|_{L^{p_i}(\X)}.
\end{talign*}
\end{lemma}

We now provide an intermediate result which upper bounds the norm of the composition of two functions. For this, we will need to introduce an ordering on multi-indices. Let $a,b$ be two multi-indices, then $a \prec b$ means that $|a| < |b|$ or, $|a| = |b|$ and $a_i < b_i$ for the smallest $i$ such that $a_i \neq b_i$. The proof closely follows \cite{Basu2016}, but allows for additional smoothness of $g\circ h$.

\begin{theorem}\label{thm:general_Sbolev_composition}
Let $\X \subset \mathbb{R}^d$ be an open set and let $\Hk$ be an RKHS with kernel $k:\R^d \times \R^d \rightarrow \R$ satisfying $k \in \mathcal{C}^{s,s}(\R^d \times \R^d)$ with $\sup_{x\in \X}\partial^{t,t}k(x,x)<C_{k},\forall t\in\mathbb{N}_0^{d}$ such that $|t|\leq s$ where $C_k$ is some universal constant only depending on kernel. Suppose $g:\X \rightarrow \R$ satisfies $g\in \Hk$ and $h:[0,1]^s \rightarrow \X$. Then, assuming $h$ is sufficiently regular for all norms to exist:
\begin{talign*}
    \VHK(g \circ h)
    & \leq C \|g\|_{\Hk(\X)} \\
    & \quad \times \sum_{\alpha \neq \varnothing, \alpha \subseteq 1:s} \sum_{1\leq\left|t\right|\leq\left|\alpha\right|} \sum_{l=1}^{\left|\alpha\right|}\sum_{\left(\ell_{r},k_{r}\right)\in S(l,\alpha,t)} \prod_{r=1}^l  \big\| \partial_{\ell_{r}} h_{k_{r}}(\cdot:1_{-\alpha}) \big\|_{L^{p_{r}}([0,1]^{|\alpha|})} 
\end{talign*}
for any  $\sum_{r=1}^{s}p_{r}^{-1} \leq 1$ and where
\begin{talign*}
    S(l,\alpha,t) 
    & 
    = \Big\{(\ell_r,k_r), r=1,\ldots, l \; \Big| \; \ell_r \in 1:s, k_r \in 1:d, \cup_{r=1}^l \ell_r = \alpha, \\  
    &\qquad \qquad \; \ell_r \cap \ell_{r'} = \emptyset \text{ for } r \neq r', \text{ and } | \{ j \in 1:l \; | \;  k_j =i\}| = t_i \Big\}. \nonumber
\end{talign*}
\end{theorem}

\begin{proof}
Starting with Equation 3 in \cite{Basu2016} and recalling that $h:[0,1]^s\rightarrow \X$ and $g:\X \rightarrow \R$: 
\begin{talign}
\VHK(g \circ h) \leq \sum_{\alpha \neq \varnothing,\alpha \subseteq 1:s} \| \partial_\alpha (g \circ h)(\cdot:1_{-\alpha})\|_{L^1([0,1]^{|\alpha|})}\label{eq:preliminary_0}
\end{talign}
where we recall that $\VHK(g \circ h)$ denotes the variation of $g \circ h$ in the sense of Hardy and Krause. In order to express the norm of $g \circ h$, we first need an expression for its partial derivatives. We will use Theorem 1 in \cite{Constantine1996} which a Faa di Bruno formula for mixed partial derivative. In particular, for $\alpha \subseteq 1:s$:
\begin{talign*}
 \partial_{\alpha}(g\circ h)(u:1_{-\alpha})
    &
    = \sum_{1\leq\left|t\right|\leq\left|\alpha\right|}
    \partial^t g (h(u:1_{-\alpha})) \sum_{l=1}^{\left|\alpha\right|}\sum_{\left(\ell_{r},k_{r}\right)\in S(l,\alpha,t)}\prod_{r=1}^{l}  \partial_{\ell_{r}}h_{k_{r}}(u:1_{-\alpha})
\end{talign*}
To clarify, here the first sum is over all multi-indices $t\in\mathbb{N}_{0}^{d}$, and $\partial_{\alpha} (g \circ h)$ denotes mixed partial derivatives where we differentiate at most once per coordinate. 
Taking the $L^p$ norm of these derivatives, we get that for $\alpha \subseteq 1:s$:
\begin{talign}
    & \left\| \partial_{\alpha}(g\circ h)(\cdot:1_{-\alpha}) \right\|_{L^p([0,1]^{|\alpha|})} \nonumber \\
    & =  \left\|
    \sum_{1\leq\left|t\right|\leq\left|\alpha\right|}
    \partial^t g(h(\cdot:1_{-\alpha})) \sum_{l=1}^{\left|\alpha\right|} \sum_{\left(\ell_{r},k_{r}\right)\in S(l,\alpha,t)} \prod_{r=1}^{l}  \partial_{\ell_{r}}h_{k_{r}}(\cdot:1_{-\alpha})\right\|_{L^p([0,1]^{|\alpha|})} \nonumber \\
     & \leq 
    \sum_{1\leq\left|t\right|\leq\left|\alpha\right|}\left\|
    \partial^t g(h(\cdot:1_{-\alpha})) \sum_{l=1}^{\left|\alpha\right|} \sum_{\left(\ell_{r},k_{r}\right)\in S(l,\alpha,t)} \prod_{r=1}^{l} \partial_{\ell_{r}}h_{k_{r}}(\cdot:1_{-\alpha}) \right\|_{L^p([0,1]^{|\alpha|})} \nonumber \\
   & \leq 
    \sum_{1\leq\left|t\right|\leq\left|\alpha\right|} 
    \left\|\partial^t g(h(\cdot:1_{-\alpha})) \right\|_{L^\infty([0,1]^{|\alpha|})}  \left\| \sum_{l=1}^{\left|\alpha\right|}\sum_{\left(\ell_{r},k_{r}\right)\in S(l,\alpha,t)} \prod_{r=1}^{l} \partial_{\ell_{r}}h_{k_{r}}(\cdot:1_{-\alpha}) \right\|_{L^p([0,1]^{|\alpha|})} \nonumber \\
    & \leq
    \sum_{1\leq\left|t\right|\leq\left|\alpha\right|}
    \left\|\partial^t g(h(\cdot:1_{-\alpha})) \right\|_{L^\infty([0,1]^{|\alpha|})} \sum_{l=1}^{\left|\alpha\right|}\sum_{\left(\ell_{r},k_{r}\right)\in S(l,\alpha,t)} \left\| \prod_{r=1}^{l}  \partial_{\ell_{r}}h_{k_{r}}(\cdot:1_{-\alpha}) \right\|_{L^p([0,1]^{|\alpha|})} \label{eq:preliminary_1}
\end{talign}
Here, the first inequality follows by the triangle inequality. The second inequality follows from H\"older's inequality (Lemma \ref{lemma:Holder} with $p_1 =\infty$ and $p_2=p$). Finally, the third inequality once again follows from the triangle inequality. The rest of the proof will consist of bounding each of the remaining norms separately.

For the first norm, we will use the fact that $g \in \Hk$ and we can therefore bound the norm of its derivatives. Since $k \in \mathcal{C}^{m \times m}(\R^d \times \R^d)$, we have $k \in \mathcal{C}^{m \times m}(\X \times \X)$. Following Corollary 4.36 of \cite{Steinwart2008}, we have that  $g\in \Hk$ implies $g \in \mathcal{C}^m(\X)$, and $\forall t \in \mathbb{N}^d_0$ with $|t|\leq m$ and $\forall x \in \X$,
\begin{talign*}
\partial^{t} g(x)  \leq \|g\|_{\Hk(\X)} (\partial^{t,t} k(x,x))^{\frac{1}{2}}
\end{talign*}

 Given the assumption that $\sup_{x\in \X}\partial^{t,t}k(x,x)\leq C_k$, we obtain the following inequality combining above results with $m=s$
\begin{align}
    \|\partial^{t} g (h(\cdot:1_{-\alpha}))\|_{L^{\infty}([0,1]^{|\alpha|})} \leq \|\partial^{t} g\|_{L^{\infty}(\X)} = \sup_{x \in \X} \partial^{t} g(x)  \leq C_k^{1/2} \|g\|_{\Hk(\X)}.\label{eq:preliminary_2}
\end{align}

For the second norm, we can once again make use of H\"older's inequality (Lemma \ref{lemma:Holder}) which guarantees that if $\partial_{\ell_r} h_{k_r}(\cdot:1_{-\alpha}) \in L^{p_{r}}([0,1]^{|\alpha|})$ for $\alpha \subseteq 1:s$ and $\sum_{r=1}^{l}p_{r}^{-1}\leq p^{-1}$, then 
\begin{talign}
     \left\|\prod_{r=1}^{l} \partial_{\ell_{r}}h_{k_{r}}(\cdot:1_{-\alpha})\right\|_ {L^{p}([0,1]^{|\alpha|})} 
     &\leq 
     \prod_{r=1}^l  \left\|\partial_{\ell_{r}}h_{k_{r}}(\cdot:1_{-\alpha}) \right\|_ {L^{p_{r}}([0,1]^{|\alpha|})}. \label{eq:preliminary_3}
\end{talign}
Plugging the inequalities in \ref{eq:preliminary_2} and \ref{eq:preliminary_3} into \ref{eq:preliminary_1}, we get that for $\alpha \subseteq 1:s$:
\begin{talign*}
    &\|\partial_\alpha (g \circ h)(\cdot:1_{-\alpha})\|_{L^{p}([0,1]^{|\alpha|})} \\
    & \leq C \|g\|_{\Hk(\X)} \sum_{1\leq\left|t\right|\leq\left|\alpha\right|}  \sum_{l=1}^{\left|\alpha\right|}\sum_{\left(\ell_{r},k_{r}\right)\in S(l,\alpha,t)} \prod_{r=1}^l  \left\|\partial_{\ell_{r}}h_{k_{r}}(\cdot:1_{-\alpha}) \right\|_ {L^{p_{r}}([0,1]^{|\alpha|})}.
\end{talign*}
Plugging this bound with $p=1$ in Equation \ref{eq:preliminary_0} concludes the proof.
\end{proof}


\subsection{Proof of Theorem \ref{thm:QMC_concentration_ineq} and Corollary \ref{cor:MMD}}\label{appendix:proof_MMD}

\subsubsection{Proof of Theorem \ref{thm:QMC_concentration_ineq}}
\begin{proof}
First, we notice that under our assumptions, we may directly apply Theorem \ref{thm:general_Sbolev_composition} in order to get that $\exists C_\theta>0$ such that for any $f \in \Hk$:
\begin{talign*}
\VHK(f \circ G_\theta) \leq C_\theta \|f\|_{\Hk}.
\end{talign*}
More precisely, $G_\theta$ takes the place of $h$ in Theorem 4, and all norms depending on $G_\theta$ are bounded thanks to Assumption \ref{assumptions:generator}.
We can then directly combine this result with the Koksma-Hlawka inequality (Lemma \ref{lemma:KoksmaHwlaka_BVHK}) to get a bound on the MMD:
\begin{talign*}
    \mMMD(\P_\theta,\P_\theta^n) 
    & = 
    \sup_{\|f\|_{\Hk(\X)}  \leq 1}\left|\int_{\X}f(x)\P_\theta(\d x)-\int_{\X}f(x)\P^n_\theta(\d x)\right| \nonumber \\
    & = 
    \sup_{\|f\|_{\Hk(\X)}  \leq 1}\left|\int_{[0,1]^s}f(G_{\theta}(u))du-\frac{1}{n}\sum_{i=1}^{n}f(G_{\theta}(u_{i}))\right| \nonumber \\
    & \leq \sup_{\|f\|_{\Hk(\X)} \leq 1} \VHK(f \circ G_\theta) D^{*}(\{x_i\}_{i=1}^n) \nonumber \\
    & \leq \sup_{\|f\|_{\Hk(\X)} \leq 1} C_\theta \|f\|_{\Hk(\X)} D^{*}(\{x_i\}_{i=1}^n) =  C_\theta D^{*}(\{x_i\}_{i=1}^n) 
\end{talign*}
By definition, we know that whenever $\{u_i\}_{i=1}^n$ is a QMC point set, we have 
$D^*(\{u_i\}_{i=1}^n)= O(n^{-1}{\color{black}(\log{n})^{{\alpha}_s}})$.
This concludes the proof. 
\end{proof}

\subsubsection{Proof of Corollary \ref{cor:MMD}}

\begin{proof}
The proof is trivial by using the fact that $\mMMD$ is a distance and thus the triangle inequality holds $ \left|\mMMD(\P_\theta,\Q^{m})-\mMMD(\P^{n}_\theta,\Q^{m})\right|\leq\mMMD(\P_\theta, \P_\theta^{n})$. The rate therefore follows from Theorem \ref{thm:QMC_concentration_ineq}.
\end{proof}

\subsection{Proof of Theorem \ref{thm:Wasserstein_sample_complexity_QMC} and Corollary \ref{cor:wasserstein}}\label{appendix:proof_Wasserstein}

\subsubsection{Proof of Theorem \ref{thm:Wasserstein_sample_complexity_QMC}}
\begin{proof}
Using the Kantorovich-Rubinstein duality theorem, we may express the Wasserstein distance as an integral probability metric associated to the class of Lipschitz continuous functions when $p=1$:
\begin{talign}
    W_{c,1}(\P_\theta,\P_\theta^n) & :=\sup_{\|f\|_{\text{L}}\leq 1}\left|\int_\X f(x)\P_\theta(\d x)- \int_\X f(x)\P_\theta^n(\d x)\right| \nonumber \\
    & = \sup_{\|f\|_{\text{L}}\leq 1}\left|\int_{[0,1]^s} f(G_\theta(u)) du- \frac{1}{n} \sum_{i=1}^n f(G_\theta(u_i)) \right| \label{eq:Wasserstein_proof_1}
\end{talign}
where $\|g\|_{\text{L}}:=\left|g(x)-g(y)\right|/c(x,y)$. Then, using the Koksma-Hwlaka inequality in Lemma \ref{lemma:KoksmaHwlaka_BVHK}, we get:
\begin{talign}
    \left| \int_{[0,1]^s} f(G_\theta(u)) \d u - \frac{1}{n} \sum_{i=1}^n f(G_\theta(u_i)) \right| 
    & \leq V_{\text{HK}}(f \circ G_\theta ) D^*(\{u_i\}_{i=1}^n). \label{eq:Wasserstein_proof_2}
\end{talign}
Let $\mathcal{Y}^N$ be the ladder $\{ v \in [0,1]^N | 0 < v_1 < v_2< \ldots < v_N =1\}$. Assuming $f$ is Lipschitz and $G_\theta$ has bounded variation in the sense of Hardy and Krause, we have
\begin{talign}
    V_{\text{HK}}(f \circ G_\theta)
    & 
    =\sup_{N \geq 1} \sup_{v \in \mathcal{Y}^N}\sum_{i=1}^{N}\left|f(G_{\theta}(v_{i}))-f(G_{\theta}(v_{i-1}))\right| \nonumber\\
    &
    = \|f\|_{\text{L}} \sup_{N \geq 1} \sup_{v \in \mathcal{Y}^N}\sum_{i=1}^{N}   c(G_{\theta}(v_{i}),G_{\theta}(v_{i-1})) \nonumber\\
    &
    \leq M \|f\|_{\text{L}}  \sup_{N \geq 1} \sup_{v \in \mathcal{Y}^N}\sum_{i=1}^{N}\left|G_{\theta}(v_{i})-G_{\theta}(v_{i-1})\right| \nonumber\\
    &
    = M \|f\|_{\text{L}}V_{\text{HK}}(G_{\theta}). \label{eq:Wasserstein_proof_3}
\end{talign}
where the first equality follows by definition of the Hardy-Krause variation, the second equality from the definition of the Lipschitz norm, and the first inequality from the fact that all norms are equivalent on $\R$ so that $\exists M>0$ such that $c(x,y) \leq M |x-y|$ for all $x,y \in \R$. Combining the results in Equations \ref{eq:Wasserstein_proof_1}, \ref{eq:Wasserstein_proof_2} and \ref{eq:Wasserstein_proof_3}, we get:
\begin{talign*}
    W_{c,1}(\P_\theta,\P_\theta^n) 
    & = \sup_{\|f\|_{\text{L}} \leq 1} \left| \int_{[0,1]^s} f(G_\theta(u)) \d u - \frac{1}{n} \sum_{i=1}^n f(G_\theta(u_i)) \right|\\
    & \leq \sup_{\|f\|_{\text{L}} \leq 1} M \|f\|_{\text{L}} V_{\text{HK}}(G_\theta) D^*(\{u_i\}_{i=1}^n) = M V_{\text{HK}}(G_\theta) D^*(\{u_i\}_{i=1}^n).
\end{talign*}
The proof of the theorem is concluded by noting the rate of convergence for the star discrepancy in the case of QMC or RQMC point sets.
\end{proof}

\subsubsection{Proof of Corollary \ref{cor:wasserstein}}
\begin{proof}
The proof is simple by noticing when $p\geq1$, the Wasserstein distance $W_{c,p}$ with distance function $c$ is indeed a distance satisfying the triangle inequality (Proposition 2.3 in \cite{Peyre2019}). 
\end{proof}


\subsection{Proof of Theorem \ref{thm:Sinkhorn_QMC}} \label{appendix:proof_sinkhorn}

We will now prove Theorem \ref{thm:Sinkhorn_QMC}. The bounds follow the main approach in \cite{Genevay2019}, but need to be significantly modified to accommodate QMC or RQMC point sets instead of IID realisations.

\begin{proof}
Since the Sinkhorn divergence is a normalised version of the regularised optimal transport problem, we can use this definition together with the triangle inequality to get: 
\begin{talign}
& \left|S_{c,p,\lambda}(\P_{\theta},\mathbb{Q})-S_{c,p,\lambda}(\P_{\theta}^{n},\mathbb{Q})\right|  \nonumber \\
& = \left|\bar{W}_{c,p,\lambda}(\P_{\theta},\mathbb{Q})-\bar{W}_{c,p,\lambda}(\P_{\theta}^{n},\mathbb{Q}) - \frac{1}{2} \left( \bar{W}_{c,p,\lambda}(\P_\theta,\P_\theta) - \bar{W}_{c,p,\lambda}(\P_\theta^n,\P_\theta^n) \right)\right| \nonumber \\
& \leq \left|\bar{W}_{c,p,\lambda}(\P_{\theta},\mathbb{Q})-\bar{W}_{c,p,\lambda}(\P_{\theta}^{n},\mathbb{Q})\right| + \frac{1}{2} \left| \bar{W}_{c,p,\lambda}(\P_\theta,\P_\theta) - \bar{W}_{c,p,\lambda}(\P_\theta^n,\P_\theta^n) \right|. \label{eq:proof_sinkhorn_1}
\end{talign}
We will now focus on bounding these terms. To do so, we first recall that the regularized optimal transport problem can be expressed as follows:
\begin{talign*}
\bar{W}_{c,p,\lambda}(\P,\Q)
=\max_{g\in\mathcal{C}(\X),h\in\mathcal{C}(\X)}\mathbb{E}_{X \sim \P, Y \sim \Q}\left[F^{X,Y}_{\lambda}(g,h)\right]+\lambda,
\end{talign*}
where
\begin{talign*}
F_{\lambda}^{x,y}(g,h)=g(x)+h(y)-\lambda\exp \left({\frac{g(x)+h(y)-c^p(x,y)}{\lambda}}\right).
\end{talign*}
We will denote by $(g^{*},h^{*})$ the optimal potentials for $\bar{W}_{c,p,\lambda}(\P_{\theta},\mathbb{Q})$ (i.e. the functions $g$ and $h$ attaining the maximum), by $\left(\Bar{g},\Bar{h}\right)$ the optimal potentials for $\bar{W}_{c,p,\lambda}(\P_{\theta}^{n},\mathbb{Q})$, by $(\tilde{g},\tilde{h})$ the optimal potentials for $\bar{W}_{c,p,\lambda}(\P_{\theta},\P_{\theta})$, and by $(\hat{g},\hat{h})$ the optimal potentials for $\bar{W}_{c,p,\lambda}(\P^n_{\theta},\P^n_{\theta})$.

Now we can upper bound the first term in \eqref{eq:proof_sinkhorn_1} using the triangle inequality as follows:
\begin{talign}
     & \left|\bar{W}_{c,p,\lambda}(\P_{\theta},\mathbb{Q}) - \bar{W}_{c,p,\lambda}(\P_{\theta}^{n},\mathbb{Q})\right| \nonumber \\
     & =
    \left|\mathbb{E}_{X\sim\P_\theta,Y\sim\Q}\left[F_{\lambda}^{X,Y}\left(g^{*},h^{*}\right)\right]- \frac{1}{n}\sum_{i=1}^{n}\mathbb{E}_{Y\sim\Q}\left[F_{\lambda}^{x_{i},Y}\left(\Bar{g},\Bar{h}\right)\right]\right| \nonumber\\
    &\; \leq 
    \left|\mathbb{E}_{X\sim\P_\theta,Y\sim\Q}\left[F_{\lambda}^{X,Y}\left(g^{*},h^{*}\right)\right]-\mathbb{E}_{X\sim\P_\theta,Y\sim\Q}\left[F_{\lambda}^{X,Y}\left(\Bar{g},\Bar{h}\right)\right]\right| \nonumber\\
    &
    \quad +\left|\mathbb{E}_{X\sim\P_\theta,Y\sim\Q}\left[F_{\lambda}^{X,Y}\left(\Bar{g},\Bar{h}\right)\right] - \frac{1}{n}\sum_{i=1}^{n}\mathbb{E}_{Y\sim\Q}\left[F_{\lambda}^{x_{i},Y}\left(\Bar{g},\Bar{h}\right)\right]\right| \label{eq:proof_sinkhorn_2}
\end{talign}
We will now bound the remaining terms. The first term of \eqref{eq:proof_sinkhorn_2} can be upper bounded by:
\begin{talign}
    & \left| \mathbb{E}_{X\sim\P_\theta,Y\sim\Q} \left[F_{\lambda}^{X,Y}(g^{*},h^{*})\right] - \mathbb{E}_{X\sim\P_\theta,Y\sim\Q}\left[F_{\lambda}^{X,Y}(\Bar{g},\Bar{h})\right] \right| \nonumber\\
    & = \; \Big| \mathbb{E}_{X\sim\P_\theta,Y\sim\Q} \left[F_{\lambda}^{X,Y}(g^{*},h^{*})\right] - \frac{1}{n}\sum_{i=1}^{n} \mathbb{E}_{Y\sim\Q} \left[ F_{\lambda}^{x_i,Y}(g^{*},h^{*})\right]
    + \frac{1}{n}\sum_{i=1}^{n}\mathbb{E}_{Y\sim\Q} \left[F_{\lambda}^{x_i,Y}(g^{*},h^{*})\right] \nonumber   \\
    &
    \quad -\frac{1}{n}\sum_{i=1}^{n}\mathbb{E}_{Y\sim\Q} \left[F_{\lambda}^{x_i,Y}(\Bar{g},\Bar{h})\right] + \frac{1}{n}\sum_{i=1}^{n} \mathbb{E}_{Y\sim\Q} \left[F_{\lambda}^{x_i,Y}(\Bar{g},\Bar{h})\right] - \mathbb{E}_{X\sim\P_\theta,Y\sim\Q} \left[F_{\lambda}^{X,Y}(\Bar{g},\Bar{h})\right] \Big| \nonumber \\
    & \leq \; \left| \mathbb{E}_{X\sim\P_\theta,Y\sim\Q} \left[F_{\lambda}^{X,Y}(g^{*},h^{*})\right] - \frac{1}{n}\sum_{i=1}^{n} \mathbb{E}_{Y\sim\Q} \left[ F_{\lambda}^{x_i,Y}(g^{*},h^{*})\right]\right|
     \nonumber   \\
    &
    \quad + \left|\frac{1}{n}\sum_{i=1}^{n} \mathbb{E}_{Y\sim\Q} \left[F_{\lambda}^{x_i,Y}(\Bar{g},\Bar{h})\right] - \mathbb{E}_{X\sim\P_\theta,Y\sim\Q} \left[F_{\lambda}^{X,Y}(\Bar{g},\Bar{h})\right]\right| \label{eq:proof_sinkhorn_3}
\end{talign}
where the inequality follows from triangle inequality and the fact that
\begin{talign*}
   \frac{1}{n}\sum_{i=1}^{n}\mathbb{E}_{Y\sim\Q} \left[F_{\lambda}^{x_i,Y}(g^{*},h^{*})\right] - \frac{1}{n}\sum_{i=1}^{n}\mathbb{E}_{Y\sim\Q} \left[F_{\lambda}^{x_i,Y}(\Bar{g},\Bar{h})\right] \leq 0 
\end{talign*}
because of the optimality of $(\Bar{g},\Bar{h})$. We will now turn to the second term in \eqref{eq:proof_sinkhorn_1}, which can be similarly upper-bounded as follows:
\begin{talign}
& \left| \bar{W}_{c,p,\lambda}(\P_\theta,\P_\theta) - \bar{W}_{c,p,\lambda}(\P_\theta^n,\P_\theta^n) \right| \nonumber \\
& =
\left| \mathbb{E}_{X\sim\P_\theta,Y\sim\P_\theta} \left[F_{\lambda}^{X,Y}(\tilde{g},\tilde{h})\right] - \frac{1}{n^2}\sum_{i=1}^{n} \sum_{j=1}^n F_{\lambda}^{x_i,y_j}(\hat{g},\hat{h}) \right| \nonumber \\
& =
\Big| \mathbb{E}_{X\sim\P_\theta,Y\sim\P_\theta} \left[F_{\lambda}^{X,Y}(\tilde{g},\tilde{h})\right] 
- \frac{1}{n^2}\sum_{i=1}^{n} \sum_{j=1}^n F_{\lambda}^{x_i,y_j}(\hat{g},\hat{h}) 
- \frac{1}{n} \sum_{i=1}^n \E_{Y \sim \P_{\theta}} \left[F_\lambda^{x_i,Y}(\tilde{g},\tilde{h})\right] \nonumber \\
& \quad + \frac{1}{n} \sum_{i=1}^n \E_{Y \sim \P_{\theta}} \left[F_\lambda^{x_i,Y}(\tilde{g},\tilde{h})\right] 
+ \frac{1}{n^2}\sum_{i=1}^{n} \sum_{j=1}^n F_{\lambda}^{x_i,y_j}(\tilde{g},\tilde{h}) - \frac{1}{n^2}\sum_{i=1}^{n} \sum_{j=1}^n F_{\lambda}^{x_i,y_j}(\tilde{g},\tilde{h})
\Big| \nonumber \\
& \leq \Big| \mathbb{E}_{X\sim\P_\theta,Y\sim\P_\theta} \left[F_{\lambda}^{X,Y}(\tilde{g},\tilde{h})\right] 
- \frac{1}{n} \sum_{i=1}^n \E_{Y \sim \P_{\theta}} \left[F_\lambda^{x_i,Y}(\tilde{g},\tilde{h})\right] \Big| \nonumber \\
& \quad + \Big| \frac{1}{n} \sum_{i=1}^n \E_{Y \sim \P_{\theta}} \left[F_\lambda^{x_i,Y}(\tilde{g},\tilde{h})\right]  - \frac{1}{n^2}\sum_{i=1}^{n} \sum_{j=1}^n F_{\lambda}^{x_i,y_j}(\tilde{g},\tilde{h})
\Big|\label{eq:proof_sinkhorn_4}
\end{talign}
where the last inequality holds since 
\begin{talign*}
\frac{1}{n^2}\sum_{i=1}^{n} \sum_{j=1}^n F_{\lambda}^{x_i,y_j}(\tilde{g},\tilde{h})- \frac{1}{n^2}\sum_{i=1}^{n} \sum_{j=1}^n F_{\lambda}^{x_i,y_j}(\hat{g},\hat{h}) \leq 0
\end{talign*} 
due to the definition of $\tilde{g}, \tilde{h}, \hat{g}$ and $\hat{h}$.

Combining \eqref{eq:proof_sinkhorn_2}, \eqref{eq:proof_sinkhorn_3} and \eqref{eq:proof_sinkhorn_4}, we end up with several terms which take the form of absolute integration errors for integrating against $\P_\theta$ for various choices of potentials. From Theorem 2 in \cite{Genevay2019}, we know that if $c\in\mathcal{C}^{\infty, \infty}(\X \times \X)$, then all of these potentials are in $W^{m,2}(\mathcal{X})$ for $m=d/2+1$. 
We will now obtain an upper bound on the integration error for any arbitrary potentials $g,h\in W^{m,2}(\X)$. Firstly, using the definition of $F_\lambda^{x,y}(g,h)$ and the triangle inequality:
\begin{talign}
    &\left|\mathbb{E}_{X\sim\P_\theta,Y\sim\Q}\left[F_{\lambda}^{X,Y}(g,h)\right] 
    - \frac{1}{n}\sum_{i=1}^{n} \mathbb{E}_{Y\sim\Q}\left[F_{\lambda}^{x_{i},Y}(g,h)\right]\right| \nonumber \\
    &
    \leq \left|\int_{\X}g(x)\P_\theta(\d x)-\frac{1}{n}\sum_{i=1}^{n}g(x_{i})\right| \nonumber \\
    & 
    \; + \lambda \left|\int_{\X}\int_{\X}\exp\left(\frac{g(x)+h(y)-c^p(x,y)}{\lambda}\right)\mathbb{Q}(dy)\P_\theta(\d x) -\frac{1}{n}\sum_{i=1}^{n}\int_{\mathcal{X}}\exp\left(\frac{g(x_{i})+h(y)-c^p(x_{i},y)}{\lambda}\right)\mathbb{Q}(dy)\right| \nonumber\\
    & =
    \left|\int_{[0,1]^s}g(G_{\theta}(u))du-\frac{1}{n}\sum_{i=1}^{n}g(G_{\theta}(u_{i}))\right| \label{eq:proof_sinkhorn_5}
    \end{talign}
where the equality holds due to the duality equation (see e.g. Equation 6 in \cite{Genevay2019}):
\begin{talign*}
\exp \left( - \frac{g(x)}{\lambda} \right) = \int_{\X} \exp \left(\frac{h(y)-c^p(x,y)}{\lambda}\right) \Q(dy).
\end{talign*}
To bound the expression above, we may then use the Koksma-Hwlaka inequality in Lemma \ref{lemma:KoksmaHwlaka_BVHK}:
\begin{talign*}
    \left|\int_{[0,1]^s}g(G_{\theta}(u))du-\frac{1}{n}\sum_{i=1}^{n}g(G_{\theta}(u_{i}))\right| & \leq
    \VHK(g\circ G_\theta)D^{*}\left(\{u_i\}_{i=1}^{n}\right).
\end{talign*}
To conclude the proof, our approach will be to upper bound $\VHK(g\circ G_\theta)$ using Theorem \ref{thm:general_Sbolev_composition} for some sufficiently smooth kernel $k$ which we will take to be  Mat\'ern kernel $k$ of smoothness $m-d/2=1$ (see Appendix \ref{appendix:uniform_gaussian} for a definition). This will require that $G_\theta$ is sufficiently regular to satisfy the assumptions in Theorem \ref{thm:general_Sbolev_composition}, but this is true thanks to Assumption \ref{assumptions:generator}. As a result $\exists C_\theta > 0$ such that $\VHK(g\circ G_\theta) \leq C_\theta \|g\|_{\Hk}$, which leads to a bound of the form:
\begin{talign} \label{eq:proof_sinkhorn_6}
\left|\int_{[0,1]^s}g(G_{\theta}(u))du-\frac{1}{n}\sum_{i=1}^{n}g(G_{\theta}(u_{i}))\right| & \leq
    C_\theta \|g\|_{\Hk} D^{*}\left(\{u_i\}_{i=1}^{n}\right).
\end{talign}
It has been proven in Theorem 2 of \cite{Genevay2018} that the potentials $g,h\in W^{m,2}(\mathcal{X})$, where $\mathcal{X}\subseteq\mathbb{R}^d$ is a compact space and $m\in\mathbb{N}$. Conveniently, when $m>d/2$, we know that $W^{m,2}(\mathcal{X})$ is norm-equivalent to the RKHS $\Hk$ with Mat\'ern kernel $k$ of smoothness $m-d/2=1$ (see Example 2.6 in \cite{kanagawa2018}), so that $\exists C_1,C_2>0$ such that: 
\begin{talign*}
    C_{1}\|g\|_{\mathcal{H}_{k}(\X)}\leq \|g\|_{W^{m,2}(\X)} \leq C_{2}\|g\|_{\mathcal{H}_{k}(\X)}
\end{talign*}
We can then combine this result with Equation \ref{eq:proof_sinkhorn_6} to get a bound of the form
\begin{talign}
\left|\int_{[0,1]^s}g(G_{\theta}(u))du-\frac{1}{n}\sum_{i=1}^{n}g(G_{\theta}(u_{i}))\right| & \leq
    \frac{C_\theta}{C_1}  \|g\|_{W^{m,2}(\X)} D^{*}\left(\{u_i\}_{i=1}^{n}\right). \label{eq:proof_sinkhorn_7}
\end{talign} 
Putting all of the pieces together we end up with 
\begin{talign*}
 \left|S_{c,p,\lambda}(\P_{\theta},\mathbb{Q})-S_{c,p,\lambda}(\P_{\theta}^{n},\mathbb{Q})\right|  
 & \leq \tilde{C}_{\theta} D^{*}\left(\{u_i\}_{i=1}^{n}\right).
\end{talign*}
where the bound follows from combining Equations \ref{eq:proof_sinkhorn_1} and \ref{eq:proof_sinkhorn_2} to obtain an upper bound in terms of integration error, then Equation \ref{eq:proof_sinkhorn_7} to upper bound such error, and finally combining all of the constants. This concludes our proof.
\end{proof}


\section{Additional Numerical Experiments} \label{appendix:numerical_experiments}

In this section, we provide additional details on the numerical experiments presented in the main text, and also complement these with additional results to provide a more complete picture of the impact of QMC and RQMC point sets. First, in Section \ref{appendix:uniform_gaussian}, we provide additional experiments on the sample complexity for the uniform and Gaussian models. Sections \ref{appendix:bivariate_beta} and \ref{appendix:gandk} then provide additional details on the experiments with the bivariate Beta and multivariate g-and-k distributions respectively.

\subsection{Uniform and Gaussian Models} \label{appendix:uniform_gaussian}

\begin{figure}[b!]
    \centering
  \includegraphics[width=\textwidth]{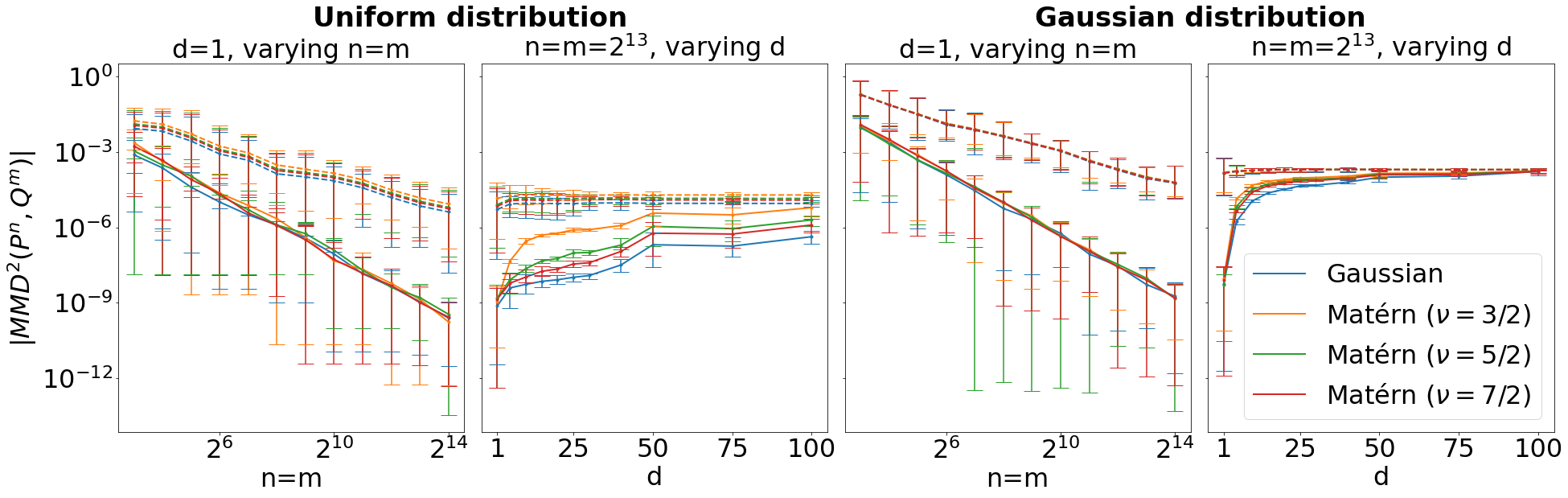}
  \caption{Sample complexity results for the MMD squared with different choices of kernels. The smoother the kernel, the more QMC point sets improve performance for $d>1$. }
\label{fig:MMD_matern_kernel}
\end{figure}

\begin{figure}[t!]
    \centering
  \includegraphics[width=\textwidth]{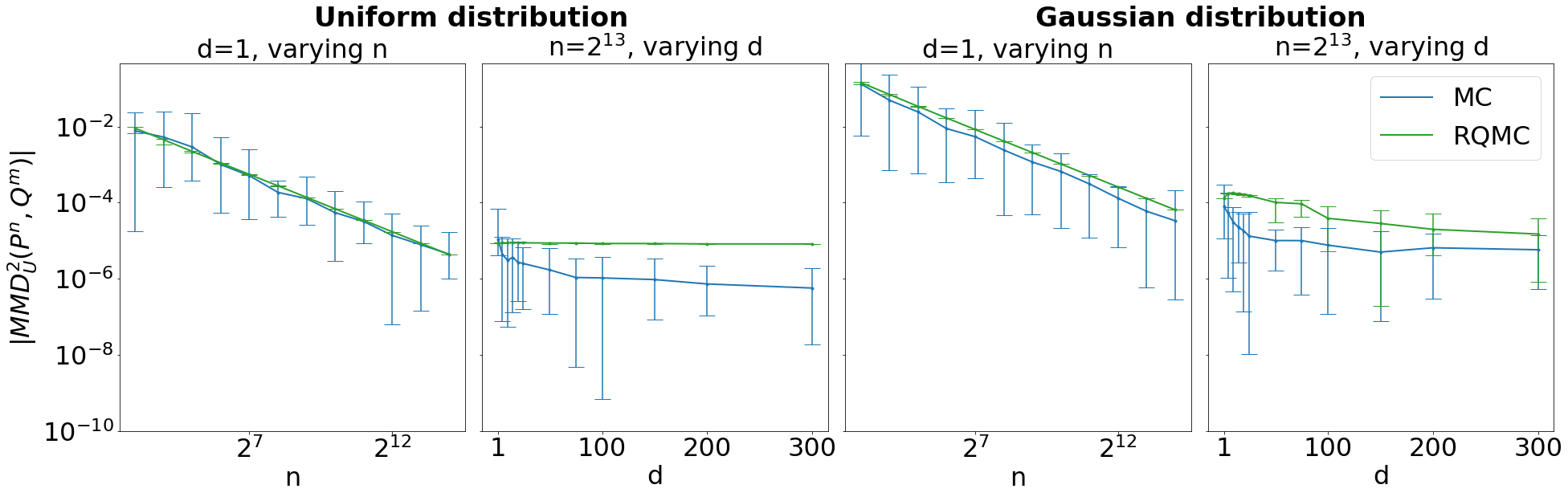}
\caption{Sample complexity results for the U-statistic approximation of the maximum mean discrepancy squared for the uniform and Gaussian distributions. We compare MC realisations with realisations obtained through a RQMC point set. The setup is identical to that of Figure \ref{fig:sample_complexity_unif_gauss} (top row), except we use a U-statistic approximation instead of the squared MMD with empirical measures.}
\label{fig:U-statistic_MMD}
\end{figure}

\begin{figure}[t!]
    \centering
  \includegraphics[width=\textwidth]{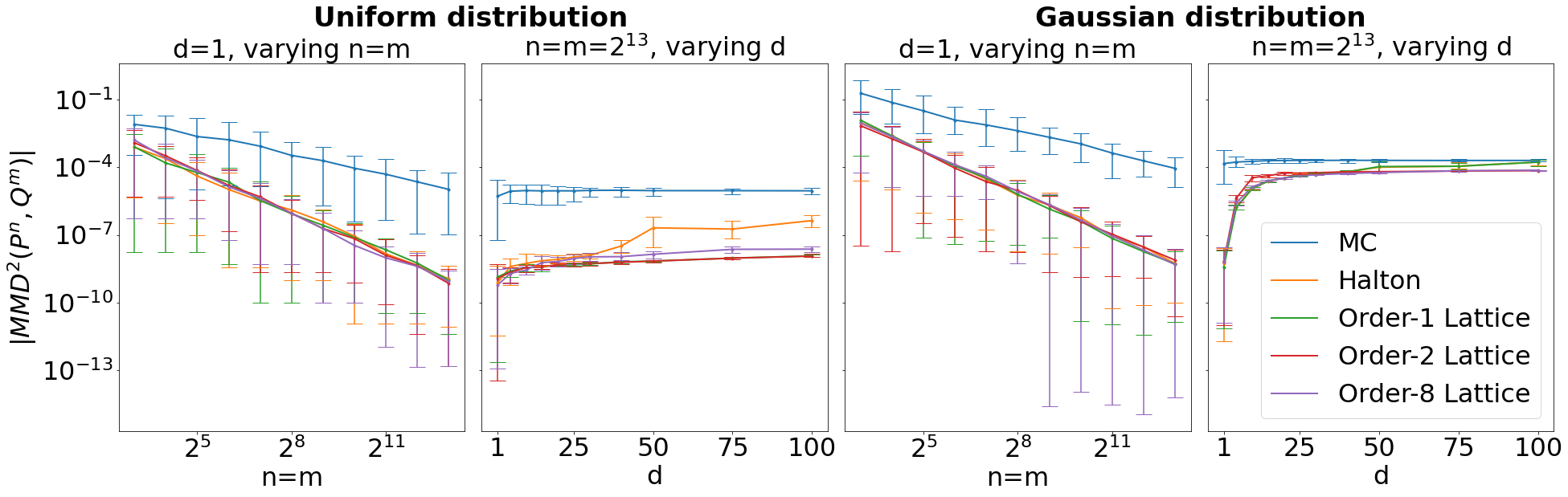}
  \caption{Sample complexity results for the MMD squared for different QMC point sets. The performance does not seem to be significantly impacted by the choice of QMC point set.}
  \label{fig:MMD_different_pointset}
\end{figure}

\begin{figure}[t!]
    \centering
  \includegraphics[width=\textwidth]{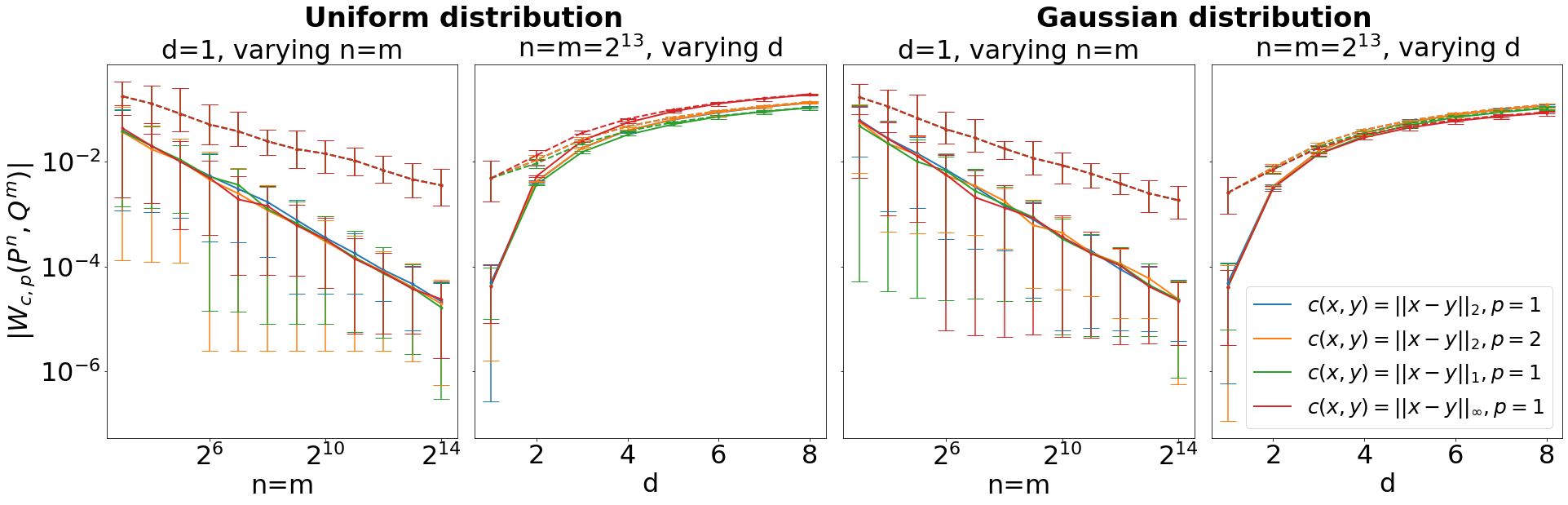}
  \caption{Sample complexity results for the Wasserstein distance with various choices of the cost $c$ and order $p$.}
  \label{fig:Wass_different_cp}
\end{figure}

\begin{figure}[t!]
    \centering
  \includegraphics[width=\textwidth]{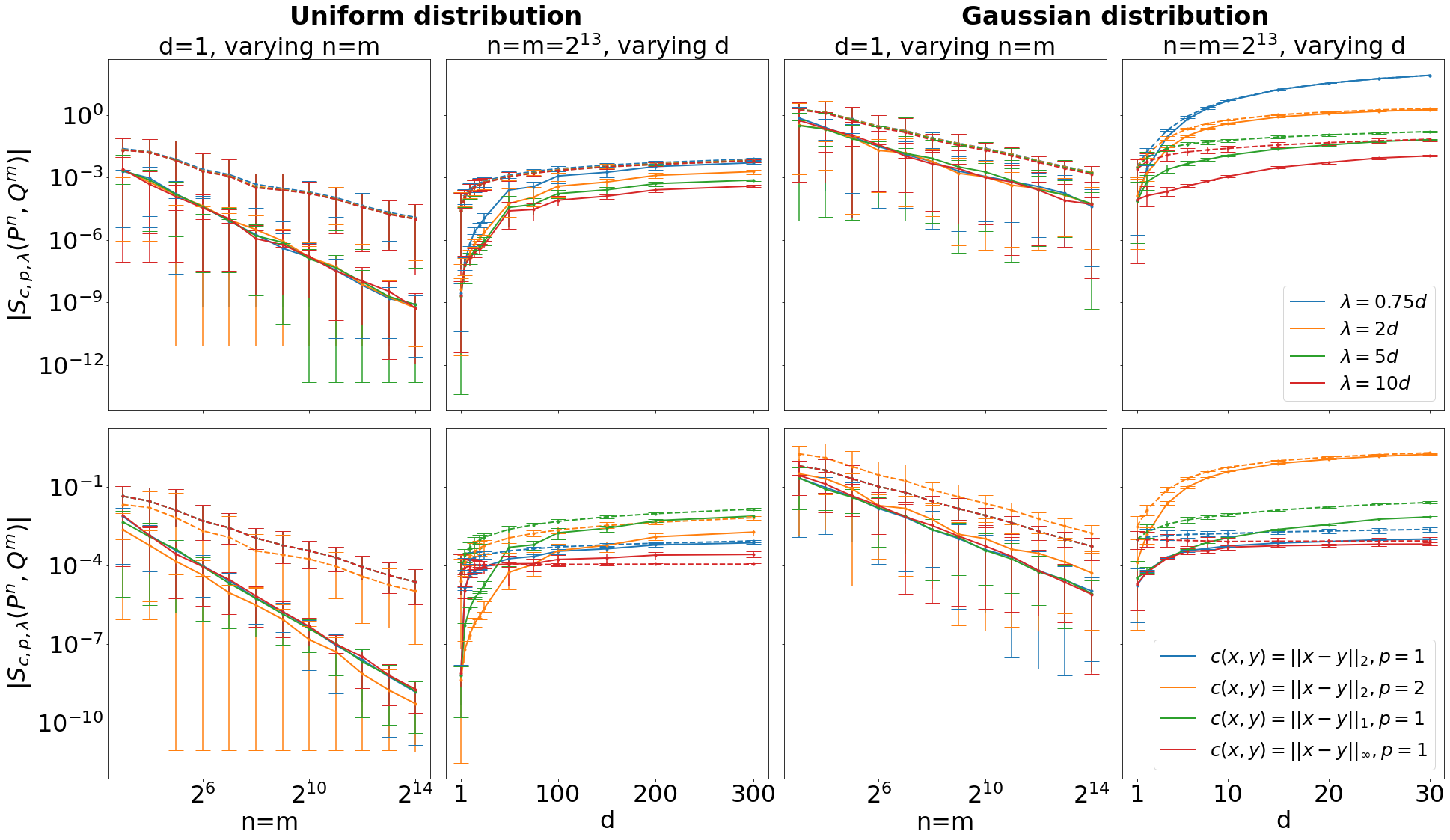}
  \caption{Sample complexity results for the Sinkhorn divergence with various choices of cost $c$, order $p$ and regularisation $\lambda$. The top row corresponds to $c=\|x-y\|_2$ and $p=2$, whereas the second row corresponds to $\lambda=2d$.}
\label{fig:Sinkorn_different_cplambda}
\end{figure}

The first set of additional experiments focuses on the sample complexity of MMD. These experiments were once again performed with generalised Halton sequences randomised using the scrambling factors of \cite{Faure2009}, and with a lengthscale of $l=1.5d^{1/2}$. In Figure \ref{fig:MMD_matern_kernel}, we compare the sample complexity of MMD when different kernels are used. In particular, we compare a squared-exponential kernel with Mat\'ern kernels of smoothness $3/2, 5/2$ and $7/2$. The Mat\'ern kernels take the form:
\begin{talign*}
k_{\nu}(x,x')=\frac{\lambda^2 2^{1-\nu}}{\Gamma(\nu)}\left(\frac{\sqrt{2\nu}\|x-x'\|_{2}}{\sigma^2}\right)^{\nu}K_{\nu}\left(\frac{\sqrt{2\nu}\|x-x'\|_{2}}{\sigma^2}\right),
\end{talign*}
where $\nu>0$ is the smoothness parameter, $\Gamma$ is the Gamma function, and $K_{\nu}$ is the modified Bessel function of the second kind of order $\nu$. In dimension $d=1$, all kernels lead to similar sample complexity results for either MC and QMC point sets. However, for $d>1$, we see a clear improvement when using a smoother kernel and QMC point sets, with the squared-exponential kernel providing the best overall performance. This clearly supports our choice of squared-exponential kernel for the experiments in the main text, and also shows the importance of the smoothness requirements on the kernel in Theorem \ref{thm:QMC_concentration_ineq}.

Another choice we made in the main text was to focus on the MMD with empirical measures. However, many papers in the literature use a U-statistic approximation instead:
\begin{talign*}
    \mMMD^{2}_\text{U}(\mathbb{P}^{n},\mathbb{Q}^{m})
    & =\frac{\sum_{i\neq j}^{n}k(x_i,x_j)}{n(n-1)}-\frac{2\sum_{i=1}^{n}\sum_{j=1}^{m}k(x_i,y_{j})}{nm}+\frac{\sum_{i\neq j}^{m}k(y_{i},y_{j})}{m(m-1)},
\end{talign*}  
see for example \cite{Briol2019,Park2016}. The main advantage of the U-statistic is that it is unbiased, but it does have a larger variance. This turns out to have a significant impact when using QMC point sets in which case we cannot obtain an improved convergence rate. This is illustrated in Figure \ref{fig:U-statistic_MMD} where we reproduced the sample complexity plots in the top row of Figure \ref{fig:sample_complexity_unif_gauss} using the U-statistic. As can be observed, we are not able to obtain a faster convergence rate, and this is the case even in $d=1$. In fact, the results are significantly worse than MC when $d>1$.

These experiments were complemented by a study of the impact of the QMC point sets in Figure \ref{fig:MMD_different_pointset} where we compare an order-1, order-2, and order-8 lattice which were shifted to obtain randomised point sets. As observed, there is only negligible differences in the performance of the different QMC point sets when $d$ is small, but further gains can be obtaioned when $d$ is large in the case of the uniform distribution. 

Next, we studied the impact of the choice of $c$ and $p$ on sample complexity results for the Wasserstein distance. Note that the result in Theorem \ref{thm:Wasserstein_sample_complexity_QMC} is only valid for $p=1$. As we can see in Figure \ref{fig:Wass_different_cp}, the performance is similar across various choices of $c$ and $p$. In each case, a faster rate is obtained for $d=1$ indicating that the result of our theorem could potentially be extended to $p \neq 1$. However, in all cases this gain in performance quickly vanishes as $d$ increases. A similar study was performed for the Sinkhorn divergence in Figure \ref{fig:Sinkorn_different_cplambda} (bottom row). In this case, we may rely on Theorem \ref{thm:Sinkhorn_QMC} which is also valid for $d>1$. As we can see, there seems to be a larger impact due to the choice of cost function or of $p$, and this should warrant further study.

\subsection{Bivariate Beta Model}\label{appendix:bivariate_beta}

As mentioned in Section \ref{sec:bivariate_beta} in the main text, it is possible to sample from the bivariate beta model using uniform random variables whenever all parameter take integer values. However, in the more general setting where the parameters may take scalar values, we will also require realisations from a Gamma random variable.

In order to make this model amenable to realisations from QMC point sets, we therefore need an approach to sampling from Gamma random variables using uniform random variables. A number of approaches are highlighted in Chapter IX.3. of \cite{Devroye1986}, but we will focus specifically on the rejection sampling algorithm by Ahrens and Dieter \cite{Ahrens1974} which we recall in Algorithm \ref{alg1}. 

\begin{algorithm}[h!] 
\caption{Rejection sampling for $n$ realisations of a $\text{Gamma}(\alpha,1)$ where $\alpha \in (0,1)$}
\label{alg1} 
\begin{algorithmic}
    \REQUIRE A 3-dimensional point set of size $n$: $\{u_i=(u_{i1},u_{i2},u_{i3})\}_{i=1}^n \subset [0,1]^3$ 
    \STATE Set $b=(\alpha+e)/e$ 
    \FOR{$i$ in $1,\ldots,n$}
        \STATE Set  $p_i = b \times u_{i1}$.
        \IF{$p\leq 1$}
            \STATE Let $x_i=p_i^{1/\alpha}$
            \STATE If $u_{i2} \leq \exp(-x_i)$, accept $x_i$, else reject $x_i$.
        \ELSE
            \STATE Let $x_i = -\log\left((b-p_i)/\alpha\right)$
            \STATE If $u_{i3} \leq x_i^{\alpha-1}$, accept $x_i$, else reject $x_i$.
        \ENDIF
    \ENDFOR
\end{algorithmic}
\end{algorithm}

\begin{figure}[h!]
	   \centering
	   \includegraphics[width=\textwidth]{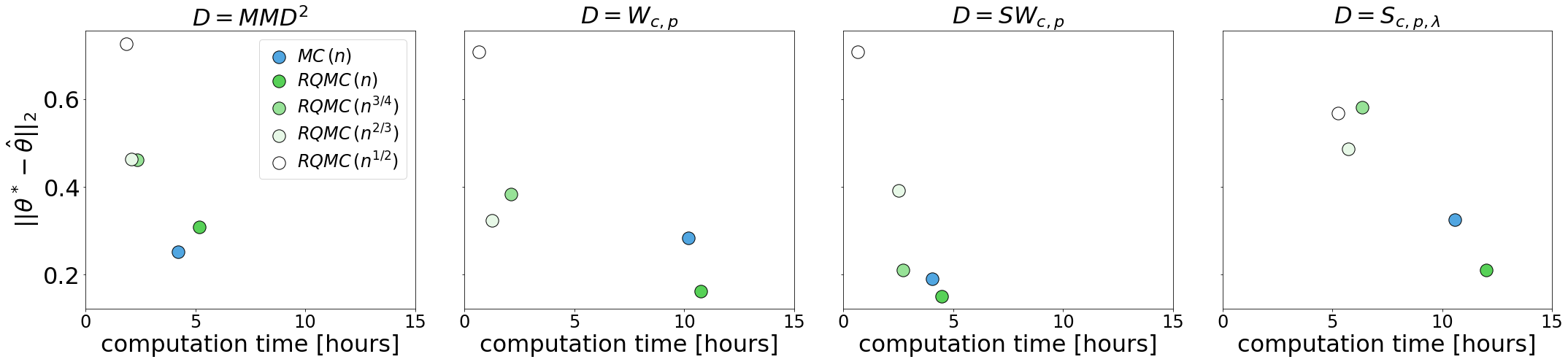}
	   \caption{\color{black}{Minimum distance estimation for the parameters of the bivariate Beta distribution with the MMD, Wasserstein, sliced Wasserstein and Sinkhorn divergence. Each point corresponds to the median of $10$ repetitions of an identical experiment.}}
	   \label{fig:bibeta_l2vstime}
\end{figure}

\newpage
{\color{black}
\paragraph{Alternative Representation} Figure \ref{fig:bibeta_l2vstime} gives an alternative representation of the results in Figure \ref{fig:MDE_bivariate_beta} in the main text. It highlights the relationship between the computation time and the $l_2$ error of the estimates in the different setups.}


\subsection{Univariate and Multivariate g-and-k Models}\label{appendix:gandk}

In this final subsection, we provide additional details for the g-and-k models.

\paragraph{Generator} In order to simulate from the multivariate g-and-k distribution studied in this paper, we will simply need to simulate some uniform random variables and transform these. In order to do so, one quantity of interest will be the matrix-square root of $\Sigma \in \R^{d\times d}$. We recall that $\Sigma$ is a symmetric tri-diagonal Toepliz matrix with diagonal entries all equal to $1$ and off-diagonal entries equal to $\theta_5$, i.e.:
\begin{talign*}
     \Sigma  = \begin{bmatrix}
    1 & \theta_5 & 0 & \dots  & 0 \\
     \theta_5 & 1 & \theta_5 & \ddots  & \vdots \\
     0 & \theta_5 & \ddots & \ddots & 0 \\
     \vdots & \ddots & \ddots & \ddots & \theta_5 \\
     0 & \dots & 0 & \theta_5  & 1
 \end{bmatrix}.
\end{talign*}
For such matrices, the square-root is known in closed form and its computation does not require the use of an algorithm. It has entries given by
\begin{talign*}
\left(\Sigma^{\frac{1}{2}}\right)_{ij} = \frac{2}{d+1} \sum_{k=1}^d \sqrt{1 +2 \theta_5 \cos\left(\frac{k \pi}{d+1}\right)} \sin \left(\frac{i k \pi}{d+1}\right) \sin \left(\frac{j k \pi}{d+1}\right).
\end{talign*}
This can be used directly in the expression for the generator of this model.

\begin{figure}[t!]
    \centering
  \includegraphics[width=0.48\textwidth]{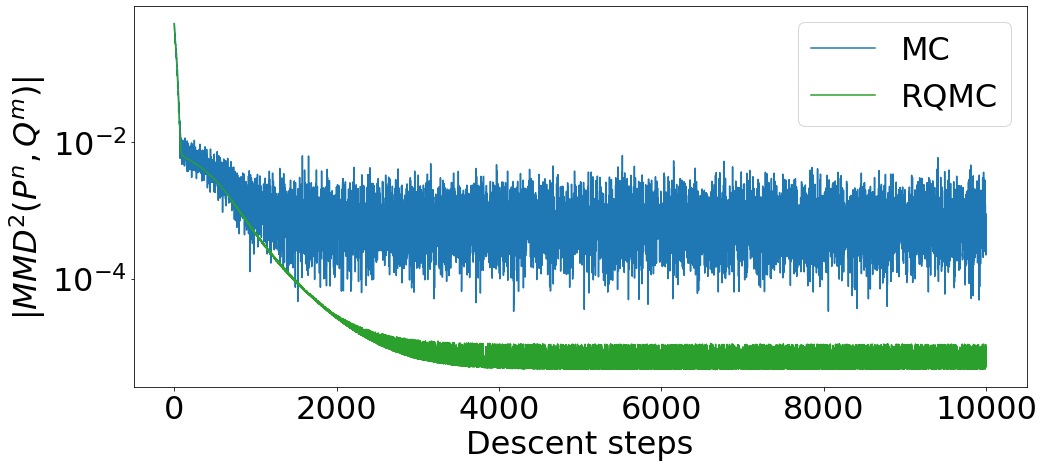}
  \includegraphics[width=0.48\textwidth]{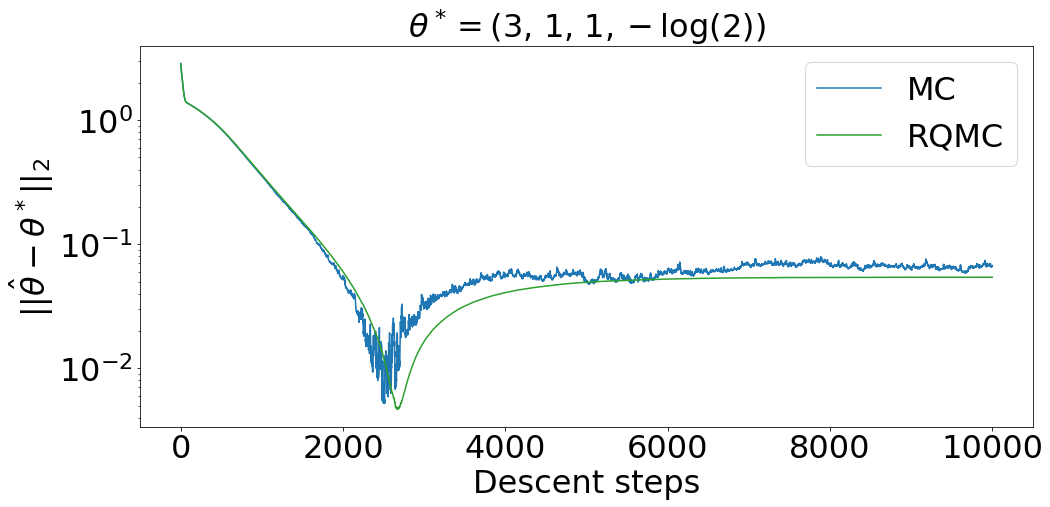}
  \includegraphics[width=\textwidth]{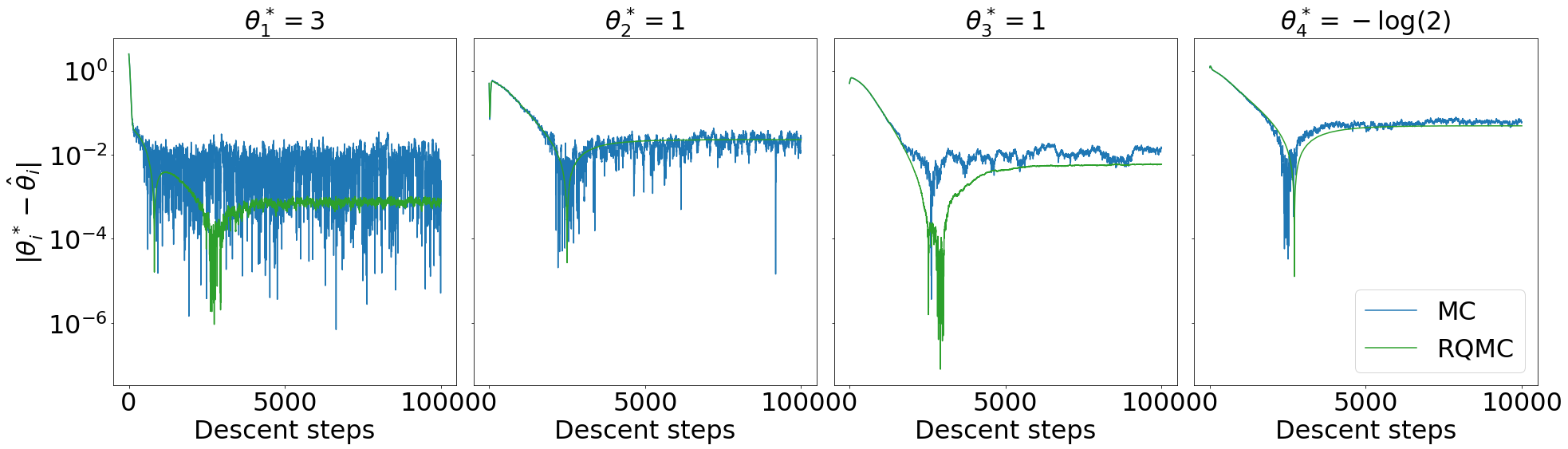}\\
  \caption{Minimum MMD estimation of the parameters of the univariate g-and-k distribution using stochastic optimisation. }
  \label{fig:univariate_results_gandk_optim}
\end{figure}
{\color{black}
\paragraph{Computational Cost} Figure \ref{fig:cost_gandk} describes the computational cost of simulating $n$ realisations of the g-and-k distribution using our implementation. In particular, it compares MC and RQMC for a range of values of $d$. When $n$ is less that $2^{12}$, the cost is usually slightly smaller with RQMC, but as $n$ goes beyond this point the cost of using MC was significantly smaller.  }

\begin{figure}[h!]
    \centering
    \includegraphics[width=0.65\textwidth]{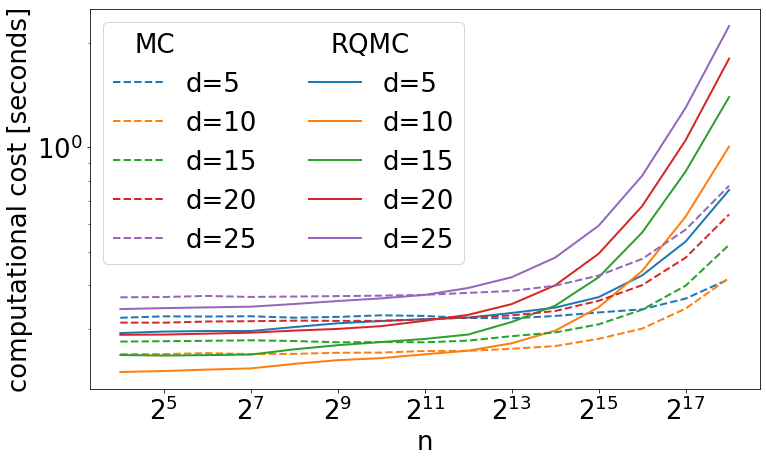}
    \caption{\color{black}{Comparison of the computational cost for simulating from the multivariate g-and-k distribution using MC and RQMC. Each line represents the average of 1000 repetitions.}}
    \label{fig:cost_gandk}
\end{figure}

\paragraph{Additional Numerical Results}
To complement the results in the main text, we first provide results for parameter estimation in the case $d=1$, which is the most common in the literature. In this case, $p=4$ since the parameter $\theta_5$ does not enter the generator.

The results were obtained without sub-sampling the dataset and are provided in Figure \ref{fig:univariate_results_gandk_optim}. As observed in the top left plot, the stochastic optimisation algorithm is able to attain low values of the MMD squared in a much smaller number of steps when using RQMC as opposed to MC. This then leads to an improved parameter estimate as measure in terms of $l_2$-norm between the estimated parameter and the true parameter $\theta^*$; see the top right plot. The bottom row of the figure gives the error for each of the four parameters as the number of step increases. In each case, the RQMC estimates provide significant improvements over the MC estimates, although the gains are limited for the second parameter (which controls the variance).

To complement these results, we provide a histogram obtained by sampling $n=2^{11}$ from the model at $\theta^*$ from MC and QMC, and compare these to a histogram of $\P_{\theta^*}$ (obtained in practice by sampling a number of samples order of magnitude larger). The results are provided in Figure \ref{fig:univariate_results_gandk_hist}. The RQMC-based realisations provide a much better approximation of the distribution near the mode. This is confirmed by the table which provides the distance between the MC-based histogram or the RQMC-based histogram and the truth in terms of various choices of distance including the Kullback-Leibler divergence, the $l_2$ norm, or the Hellinger distance. We also notice that both MC and RQMC provide relatively poor approximation at the tail of the distribution. This is most likely due to the small number of realisations used to create the histogram.

Finally, Figure \ref{fig:gandk_d2_optim} provides the $l_2$ error between true and estimated parameters for the experiment presented in Figure \ref{fig:gandk_d2_optim_MMD}. Clearly, a smaller value of the estimated MMD does not necessarily guarantee a better parameter estimate.

\begin{figure}[t!]
    \centering
  \includegraphics[width=0.7\textwidth]{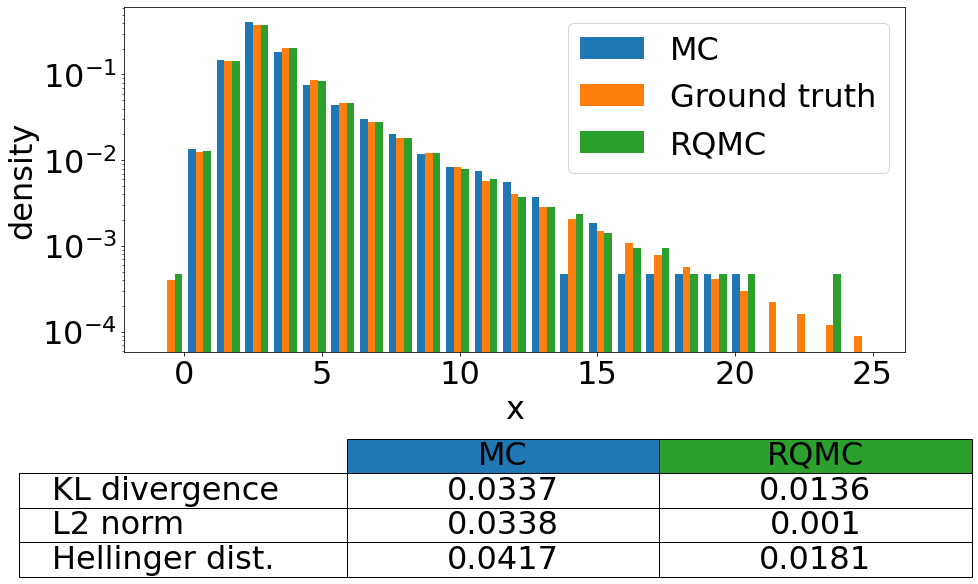}
  \caption{Histograms of the univariate g-and-k distribution at $\theta^*$, together with MC and RQMC-based approximations with $n=2^{11}$. }
  \label{fig:univariate_results_gandk_hist}
\end{figure}

\begin{figure}[t!]
    \centering
      \includegraphics[width=0.85\textwidth]{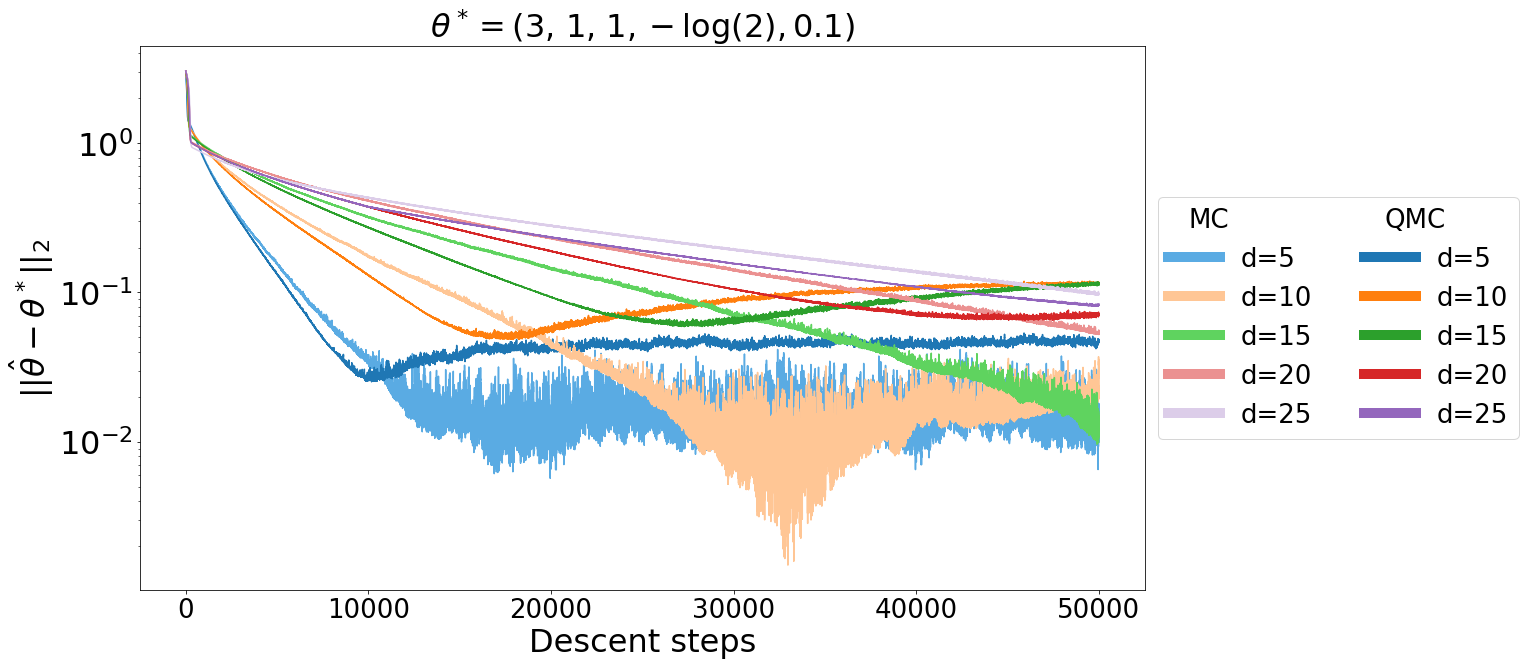}
  \caption{Minimum distance estimation with the MMD for the multivariate g-and-k distribution. The figure plots the $l_2$ error between the estimated parameter and the true value as the number of stochastic gradient descent steps increases.}
  \label{fig:gandk_d2_optim}
\end{figure}

\end{document}